\documentclass[12pt,a4paper]{article}

\usepackage{amsmath}
\usepackage{graphicx,psfrag,epsf}
\usepackage{enumerate}
\usepackage{url} 
\usepackage{ulem}
\usepackage{amssymb}
\usepackage{graphicx}
\usepackage{epsfig}
\usepackage{amsthm}
\usepackage{mathrsfs}
\usepackage{amsfonts}
\usepackage{hyperref}
\usepackage{subfigure}
\usepackage{setspace}
\usepackage{color}
\IfFileExists{url.sty}{\usepackage{url}}
{\newcommand{\url}{\text}}

\newtheorem{defn}{Definition}
\newtheorem{theorem}{Theorem}
\newtheorem{proposition}{Proposition}
\newtheorem{lemma}{Lemma}

\newtheorem{remark}{Remark}

\newtheorem*{P1*}{P. 1. (affine-invariance)}
\newtheorem*{P2*}{P. 2. (maximality at the center)}
\newtheorem*{P3*}{P. 3.  (monotonicity relative to the deepest point)}
\newtheorem*{P4*}{P. 4. (vanishing at infinity)}
\newtheorem*{P5*}{P. 5.  (continuous as a function of $x$ )}
\newtheorem*{P6*}{P. 6.  (continuous as a functional of $P$ )}

\begin{document}
	
	\begin{center}
		{\Large \textbf{An integrated local depth}}

		{Lucas Fernandez-Piana$^{\ast}$ and Marcela Svarc$^{\ast\ast}${\footnote{Corresponding author:
					Marcela Svarc, Departamento de Matem\'atica y Ciencias, Universidad de San Andr\'{e}s, Vito Dumas 248, Victoria, Argentina. Email: msvarc@udesa.edu.ar}}}

		\noindent\emph{$^{\ast}$Instituto de C\'alculo, FCEyN, Universidad
			de Buenos Aires and CONICET, Argentina.}
		\\[0pt]\emph{$^{\ast\ast}$Departamento de Matem\'atica y Ciencias,
			Universidad de San Andr\'es and CONICET, Argentina}

	\end{center}

\begin{abstract}
	We  introduce the Integrated Dual Local Depth which is a local depth measure for data in a Banach space based on the use of one-dimensional projections. The  properties of a depth measure are analyzed under this setting and a proper definition of local symmetry is given.  Moreover, strong consistency results for the local depth and also for the local depth regions are attained. Finally, applications to descriptive data analysis and classification are analyzed, making the special focus  on multivariate functional data, where we obtain very promising results.
	
\end{abstract}

\noindent%

\section{Introduction}

Data depth measures play an important role when analyzing complex data sets, such as functional or  high dimensional data. The main goal of depth measures is to provide a center-outer ordering of the data, generalizing the concept of median. Depth measures are also useful for describing different features of the underlying distribution of the data. Moreover, depth measures are powerful tools to deal with several inference problems such as, location and symmetry tests, classification, outlier detection, etc.

Nonetheless, since one of their major characteristics is that the depth values decrease along any half-line ray from the center, they are not suitable for capturing characteristics of the distribution  when data is multimodal. Hence, over the last few years, there have been introduced several definitions of local depth,  with the aim of revealing the local features of the underlying distribution. The basic idea is to restrict a global depth measure to a neighborhood of each point of the space. In this way, a local depth measure should behave as a global depth measure with respect to the  neighborhoods of the different points. It can be considered that the $h$-mode (see Cuevas et al. \cite{CFF07}), which is a depth measure defined for data in an arbitrary normed space, is the first precedent of a formal definition of local depth. This notion of depth has a ``local'' meaning since it is defined via a kernel function and recalls, in a way, the idea of defining a depth in $\mathbb{R}^d,$ in terms of the density function. Since in infinite dimensional spaces there is no natural concept of density, this definition can serve as a substitute. The definition strongly depends on the choice of a bandwidth, which remains fixed throughout the data. This definition has attract the attention in the last years, the consistency has been proved by Nagy \cite{N15}  and the classical properties for depth measures have been deeply studied by Gijbels and Nagy \cite{GN17}. Agostinelli and Romanazzi \cite{AR11}   gave the first formal definition of local depth for the case of multivariate data. They extended the concepts of simplicial and half-space depth so as to allow  recording the local space geometry near a given point. For simplicial depth, they consider only random simplices with sizes no greater than a certain threshold, while for half-space depth, the half-spaces are replaced by infinite slabs with finite width. Both definitions strongly rely on a tuning parameter, which retains a constant size neighborhood of every point of the space, something which plays an analogous role to that of  bandwidth in the problem of density estimation. Desirable statistical theoretical properties are attained for the case of univariate absolutely continuous  distributions. Paindaveine and Van Bever \cite{PB13} introduce a general procedure for multivariate data that allows converting any global depth into a local depth. The main idea of their definition is to study local environments.  This means  regarding the local depth as a global depth restricted to some neighborhood of the point of interest. They obtain strong consistency results of  the sample version with its population counterpart. Following the same line as in the case of the $h$-mode Chen et al. \cite{CDPB09}  and Sguera et al. \cite{SGL214} introduce the kernelized  spatial depth for the multivariate and functional cases respectively.
 More recently, for the case of functional data, Agostinelli \cite{A18}  gives a definition of local depth extending the ideas introduced by Lopez-Pintado and Romo \cite{LR11}  of a half-region space. This definition is also suitable for finite large dimensional datasets. Asymptotic results are obtained.
The proposals given by Agostinelli and Romanazzi \cite{AR11}, Paindaveine and Van Bever \cite{PB13} and Agostinelli \cite{A18} explicitly provide a continuum between definitions of  local and global depth.
All the authors highlight the usefulness of local depth to well known statistical problems such as: classification (Sguera et al. \cite{SGL214}, Paindaveine and Van Bever \cite{PB13}), outlier detection (Chen et al. \cite{CDPB09} and Sguera et al. \cite{SGL15}), clustering (Agostinelli \cite{A18}), regression depth (Paindaveine and Van Bever \cite{PB13}), among others.

Our goal is to give a general definition of local depth for random elements in a Banach space, extending the definition of global depth given by Cuevas and Fraiman \cite{CF09}, where they introduce the Integrated Dual Depth (IDD). The main idea  of IDD is based on combining one-dimensional projections and the notion of one-dimensional depth. Let $(\Omega, A, \mathbb{P})$ be a probability space and $\mathbb{E}$ a separable Banach space.  Denote by $\mathbb{E}'$ the dual space. Let $X:\Omega\longrightarrow \mathbb{E}$ be a random element in $\mathbb{E}$ with distribution $P$ and $Q$ a probability measure in $\mathbb{E}'$ independent of $P.$  The IDD is defined as,
\begin{equation}\label{IDD}
IDD(x,P) = \int D(f(x),P_f) dQ(f),
\end{equation}
where $D$ is an  univariate  depth (for instance,  simplicial or Tukey depth), $f \in \mathbb{E}',$ $x \in \mathbb{E}$ and $P_f$ is the univariate distribution of $f(X).$

 In the present paper, we define the Integrated Dual Local Depth (IDLD). The main idea is to replace the global depth measure in Equation (\ref{IDD}) with a local one-dimensional depth measure following the definition given in Paindaveine and Van Bever \cite{PB13}. We study how the classical properties, introduced by Zuo and Serfling \cite{ZS00},  should be analyzed within the framework of local depth.  We prove, under mild regularity conditions, that our proposal enjoys those properties. Moreover,  uniform strong consistency  results are exhibited for the definition of the empirical local depth of the population counterpart, and also for the local depth regions. The main advantages of our proposals are its flexibility in dealing with general data and also its low computational cost, which enables it to work with high-dimensional data.
We apply local depth to classification problems, with special emphasis on the application of multivariate functional data.

The remainder of the paper is organized as follows. In Section 2 we define the integrated dual local depth and study its basic properties. Section 3 is devoted to the asymptotic study of the proposed local depth measure.  In Section 4 the local depth regions are defined and the consistency results are exhibited. A simulation study and real data examples are given in Section 5. Some concluding remarks are given in Section 6. All the proofs appear in the Appendix.

\section{General Framework and Definitions}\label{RandomLocal}

In this section, we first review the concept of local depth for the univariate case. Then we  define  the Integrated Dual Local Depth, and we finally show that, under mild regularity assumptions, our proposal has good theoretical properties that correspond to those established in Paindaveine and Van Bever \cite{PB13}.
Let $P^{1}$ be a probability measure on $\mathbb{R}$ and $z\in{\mathbb{R}}.$ Let $LD(z,P^{1})$ be a local depth measure of $z$ with respect to $P^{1}$, for example, the univariate simplicial depth, that is
\begin{equation}
LD_S^{\beta}(z,P^{1})=\frac{2}{\beta^2}\left(F(z+\lambda^{\beta}_z)-F(z) \right) \left(F(z)-F(z-\lambda^{\beta}_z)\right),
\label{profsimplocalunidim}
\end{equation}
 where $F$ is the cumulative distribution function of $P^{1}$ and $\lambda^{\beta}_z$ is the neighborhood width defined as follows.
 \begin{defn}\label{localityparamdef}
 Let $F$ be a univariate cumulative distribution function and $z \in \mathbb{R}.$ Then, for $\beta \in (0,1],$ we define the neighborhood width $\lambda^{\beta}_z$ by
 \begin{equation}
 \lambda^{\beta}_z=\inf{ \left\{\lambda>0 : F(z+\lambda)-F(z-\lambda) \geq \beta \right\}},
 \label{localityparam}
 \end{equation}
 where $\beta$ is the locality level.
 \end{defn}
 \begin{remark}
If $F$ is absolutely continuous, the infimum in Equation (\ref{localityparam}) is attained and hence,
  $
 \lambda^{\beta}_z=\min{ \left\{ \lambda>0 : F(z+\lambda)-F(z-\lambda) \geq \beta\right\} }.
 $
 Even more, it is clear that if $\beta_1 < \beta_2,$ then $\lambda^{\beta_1}_z < \lambda^{\beta_2}_z.$
 \end{remark}
The locality level $\beta$ is a tuning parameter that determines the centralness of the point $z$ of the space conditional to a given window around $z.$ If the value is high it approaches the regular value of the point depth whereas if it is low it will only describe the centralness in a small neighborhood of $z$. As $\beta$ tends to one, the local depth measure tends to the global depth measure.

We can also define, in an analogous way, the Tukey univariate local depth,
 \begin{equation*}
LD_H^{\beta}(z,P^{1})=\frac{1}{\beta}min{ \left\{F(z+\lambda^{\beta}_z)-F(z),F(z)-F(z-\lambda^{\beta}_z)\right\} }.
\end{equation*}

In what follows, without loss of generality, we restrict our attention to the case of simplicial local depth, $LD_S^{\beta}$.

\subsection{Integrated Dual Local Depth} \label{IDLDsection}

Our aim in this section is to extend the IDD introduced by \cite{CF09}, to the  local setting. The IDD is a depth measure defined for random elements in a general Banach space. The  idea is to project the data according to random directions and compute the univariate depth measure of the projected unidimensional data. To obtain a global depth measure, these univariate depths measures are integrated. Under mild regularity conditions, the IDD satisfies the basic properties of depth measures described by Zuo and Serfling \cite{ZS00}, and it is strongly consistent. In addition, it is important to highlight its low computational cost, even in high dimensions, since it is based on the repeated calculation of one-dimensional projections (the Appendix includes a numerical analysis to support this claim).

Let $\Omega$ be a probability space and $\mathbb{E}$ a separable Banach space, with $\mathbb{E}'$ its  dual space. Let $X:\Omega\longrightarrow \mathbb{E}$ be a random element in $\mathbb{E}$ with distribution $P,$  $Q$ a probability measure in $\mathbb{E}'$ independent of $P$, $\beta \in (0,1],$ and $x \in \mathbb{E}$. We define the Integrated Dual Local Depth (IDLD),
\begin{equation}\label{IDLD}
IDLD^{\beta}(x,P) = \int LD_{S}^{\beta}(f(x),P_f) dQ(f),
\end{equation}
where $LD_{S}^{\beta}$ is the univariate local depth given in Equation (\ref{profsimplocalunidim}), $f \in \mathbb{E}',$ $x \in \mathbb{E}$ and $P_f$ is the univariate distribution of $f(X).$ As suggested by \cite{CF09}, in the infinite dimensional setting $Q$ may be chosen to be a non-degenerate Gaussian measure and in the multivariate setting as a uniform distribution in the unitary sphere. The fact that $Q$ is independent of $P$ allows us to proceed to the univariate case when projecting $X$, where it is feasible and easy to calculate the local depth.
With a slight abuse of notation, we write $F_f = F_{f(X)}$ for the cumulative distribution function of $f(X).$ Specifically, it reduces to $F_{f(X)}(t) = P_{f(X)} \left( (- \infty,t] \right) = P(f(X) \leq t).$

From now, every time we mention that $X$ is a continuous random element we mean that for every non-null functional on $\mathbb{E}',$ $f(X)$ is an absolutely continuous random variable.
It is clear that the IDLD is well-defined, since it is bounded by $\frac{1}{2}$ and non-negative.

Local depth measure aims to highlight data local features that play an important role in statistical analysis. We consider the well-known growth data set, which is publicly available in the FDA R package. It consists of the height measurements of 54 girls and 39 boys, who were measured 31 times from 1 to 18 years. The original data was smoothed by monotonic cubic regression spline. The growth velocity has different patterns in boys and girls. At puberty, the growth rate has a local maximum. This typically occurs between 12 and 14 years for girls, a couple of years earlier than for boys, see Figure \ref{centralesgrowth}.

IDLD  is sensitive to the choice of the locality parameter, $\beta$, its choice will be addressed in detail later in this work. For different locality levels, the features that stand out from each dataset vary. In our example, for locality levels lower than $0.35,$ within the $5\%$ locally deepest curves  two belong to boys and three to girls, maintaining the proportion in which they appear in the dataset. While for locality levels greater than $0.6,$ the $5\%$ deepest curves correspond to boys, as exhibited in Figure \ref{centralesgrowth}. For locality levels between $0.40$ and $0.55$ typically only one of the  $5\%$ locally deepest curves corresponds to boys. Enlarging the proportion of deepest curves (up to $20\%$) this pattern is still observed. 

\begin{figure}[!t]
\centering
 \includegraphics[width=0.6 \textwidth]{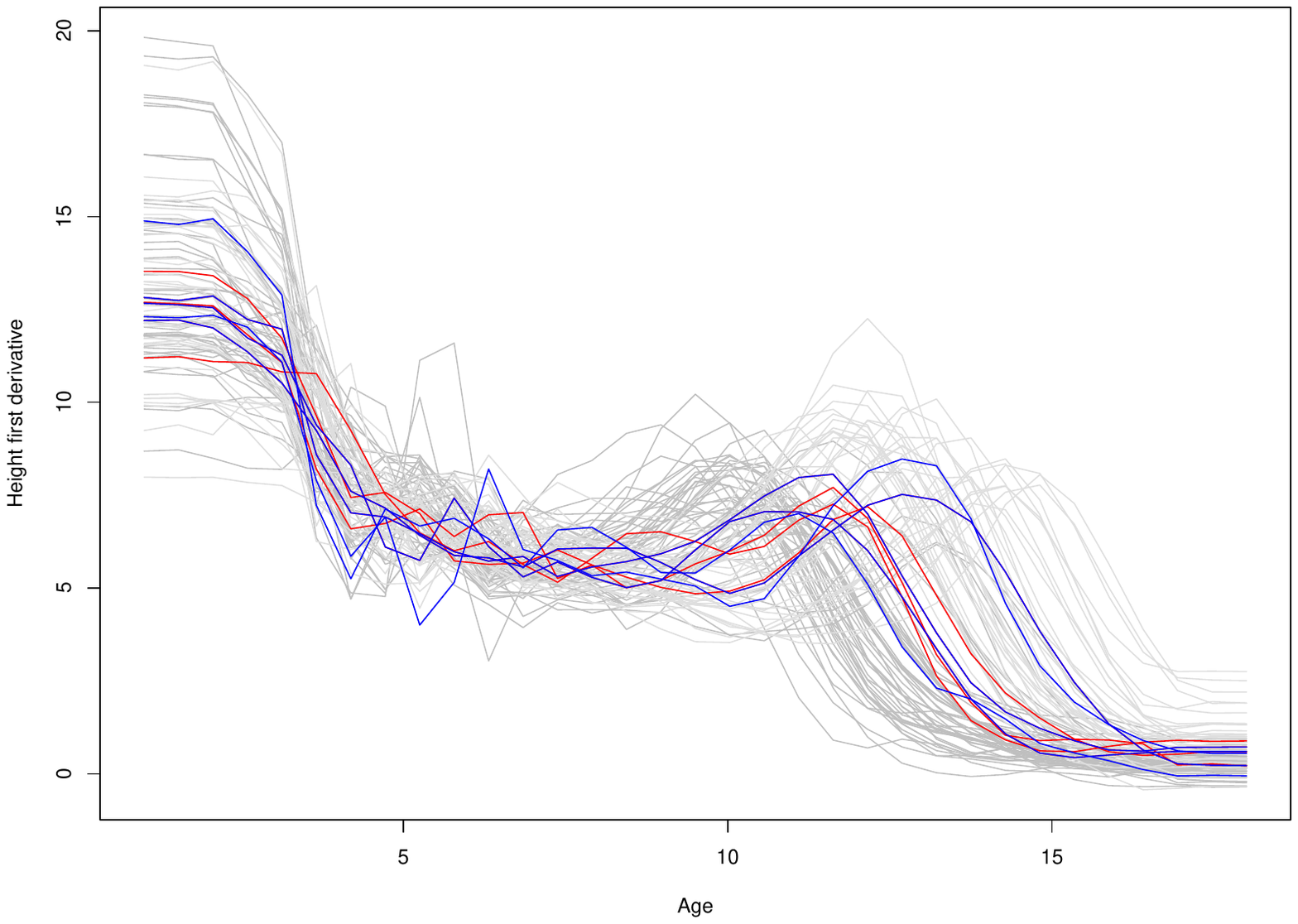} \par
\caption{ The dark grey curves show the velocity growth for boys, while the light grey for girls. The red curves are the $5\%$ global deepest curves, while the blue ones are the  $5\%$ local deepest curves for $\beta=0.2.$ } \label{centralesgrowth}
\end{figure}

Zuo and Serfling \cite{ZS00} established the general properties that depth measures should satisfy (\textbf{P. 1} - \textbf{P. 6}). Paindaveine and Van Bever \cite{PB13} analyze some of those properties for their proposal. In this work we study the properties that a local depth should fulfill, specifically, we show that under certain conditions, our proposal satisfies them.

The first property deals with the invariance of the local depths. For the finite dimensional case, IDLD is independent of the coordinate system. This property is inherited from the IDD. Since IDLD is a generalization of IDD, which  is not in  general affine invariant (i.e., let $A$ be a non-singular linear transformation in $\mathbb{R}^p$ and $P_{AX}$ denote the distribution
of $AX;$ then $D(Ax,P_{AX})$ is not equal to $D(x,P_{X})$),  neither is IDLD. It is clear that IDLD  is also invariant under translations and changes of scale.

The proofs of properties \textbf{P. 1} - \textbf{P. 6} appear in the Appendix.
\begin{P1*} Let $\mathbb{E}=\mathbb{R}^d$ with the Euclidean norm and  $X \in \mathbb{E}$ a random vector, let $Q$ be the uniform measure on the unit sphere of  $\mathbb{E}'$ independent of $P_X.$ Let $A : \mathbb{E} \rightarrow \mathbb{E}$ be an orthogonal transformation  and $\beta \in (0,1].$ Then
$IDLD^{\beta}(Ax,P_{AX}) = IDLD^{\beta}(x,P_X).$
\label{InvarianzaAfin}
\end{P1*}
\begin{remark}
 It is well known that the spatial median is not affine invariant, hence, transformation and retransformation methods have been designed
 to construct affine equivariant multivariate medians (Chakraborty and Chauduri \cite{CC96} and Chakraborty and Chauduri \cite{CC98}).
  IDLD can be modified following the ideas of  Kot\'{i}k and Hlubinka \cite{KH17} to attain this property.
\end{remark}
Depth measures are powerful analytical tools, especially in  cases where the random element enjoys symmetry properties. Local depths should locally (restricted to certain neighborhoods) inherit these properties.
Hence we  give an appropriate  definition of local symmetry.

\begin{defn}
\label{bsimetrica}

Let $Z$ be a real random variable and $\beta \in (0,1].$ Then $Z$ is said to be  $\beta$-symmetric about $\theta$ if the cumulative function distribution $F$
satisfies
\begin{equation}
 F \left( \theta + \lambda_{\theta}^{\beta'} \right) - F( \theta ) = \frac{\beta'}{2},  \mbox{ for every } 0<\beta' \leq \beta.
 \label{bsymm}
\end{equation}

 A random element $X$ in a Banach space $\mathbb{E}$ is  $\beta$-symmetric about $\theta$ if for every not null  $f \in \mathbb{E}',$  $f(X)$ is $\beta$-symmetric.
\end{defn}

The notion of $\beta$-symmetry aims to locally capture the behavior of a unimodal random variable on a neighborhood of probability $\beta,$ about $\theta,$ the locally deepest point. Figure \ref{betaSimLindo} (a) and (b) exhibit a bimodal distribution, with modes at $\theta=1$ and $\theta=4.$ On the former, both modes are  local symmetry points for $\beta=0.25$, while on the latter $\theta=4$ is a local symmetry point for $\beta=0.4$ but $\theta=1$ is not a local symmetry point for $\beta=0.4,$ the shaded area around $\theta=1$ is non-symmetrical.

\begin{figure}
\centering
\subfigure[$\theta=1$ and $\theta=4$ are local symmetry points with locality level $0.25$]{\includegraphics[width=50mm]{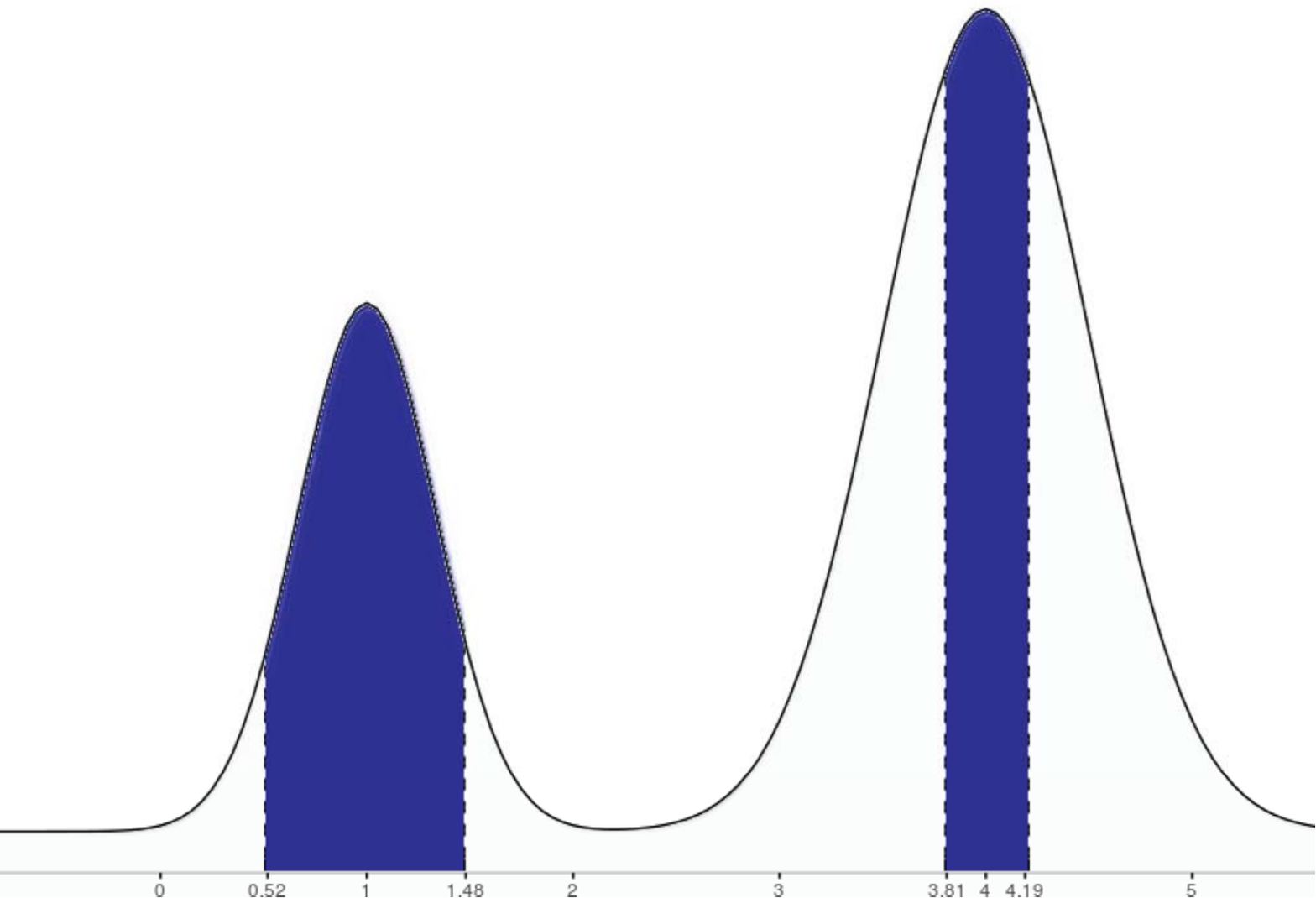}} \hspace{10mm}
\subfigure[ $\theta=4$ is a local symmetry point with locality level $0.4,$ while $\theta=1$ it it not a local symmetry point at local level $0.4$]{\includegraphics[width=50mm]{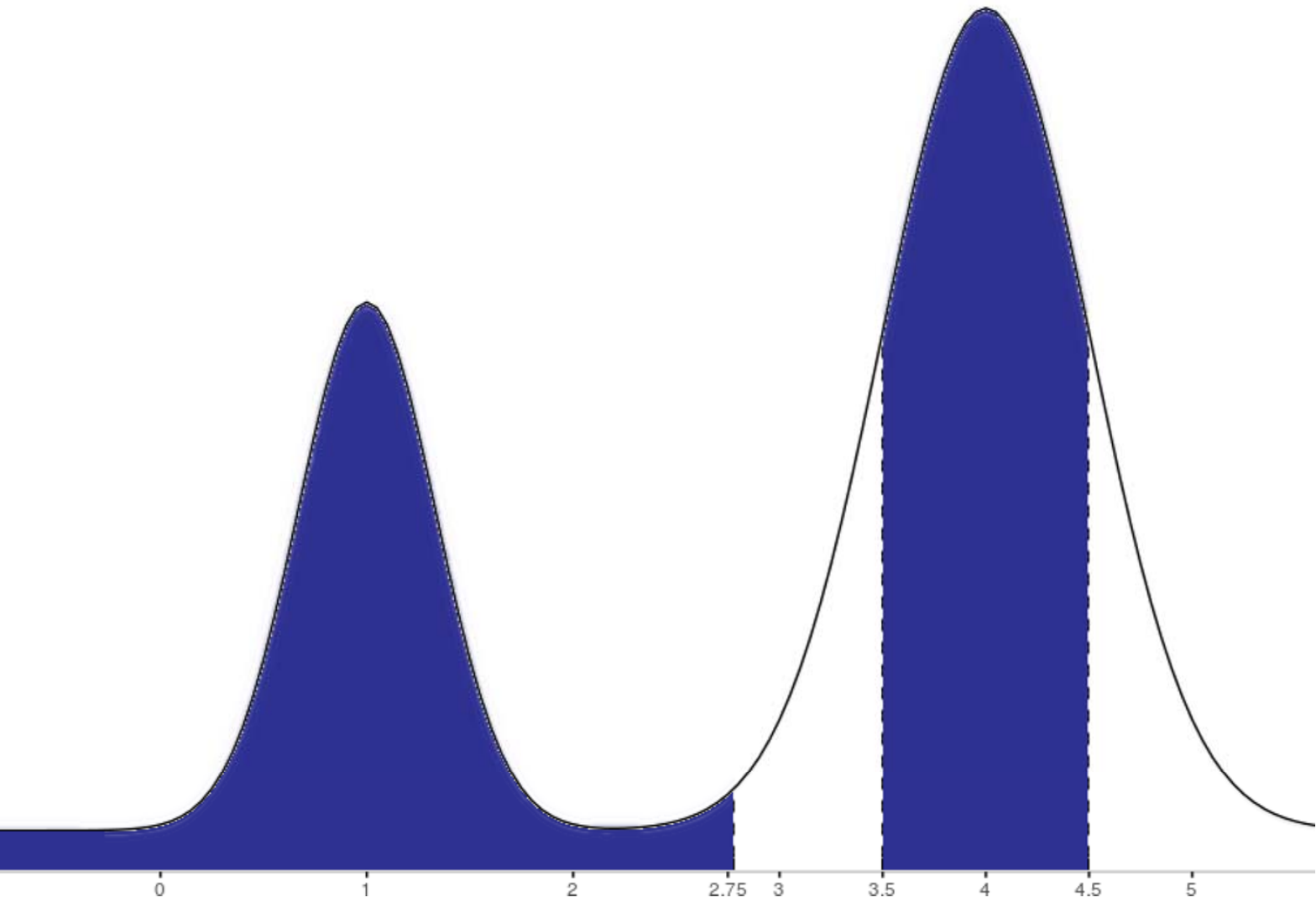}}
\caption{ Local symmetry points.} \label{betaSimLindo}
\end{figure}

An important property of depth measures is  maximality at the center, meaning that if $P$ is symmetric about  $\theta,$ then $D(x,P)$ attains its maximum value at that point. This property should be inherited by local depths if the distribution of $P$ is unimodal and convex. Local depths are relevant  for detecting local features, for instance local centers, hence our aim is to extend the property of maximality at the center  to each point  $\theta,$ that is $\beta$-symmetry.

\begin{P2*} Let $X \in \mathbb{E}$ be a random continuous element $\beta$-symmetric about $\theta.$ For $\beta \in (0,1]$ we have that
\begin{equation}
IDLD^{\beta'}(\theta,P_{X}) = \displaystyle \max_{x \in \mathbb{E}} IDLD^{\beta'}(x,P_X), \mbox{ for every } 0<\beta' \leq \beta.
\end{equation}
\label{bmaximality}
\end{P2*}

Proposition \ref{bcsymmetry} bridges the definition of $\beta$-symmetry with the usual definition of $C$-symmetry (see \cite{ZS00}).
According to Dharmadhikari and Joag-dev \cite{DJ88} a random variable $Z$ is unimodal on $z_0$ if the cumulative distribution function $F_Z$ is convex on $(-\infty, z_0)$ and it is concave  on $(z_0, \infty).$ In addition, if it is absolutely continuous, then the density function is increasing on $(-\infty, z_0)$ and decreasing on $( z_0,\infty).$ It is important to remark that in what follows we consider the continuous representantion of the density function.

\begin{proposition} Let $X \in \mathbb{E}$ be a random continuous element $C$-symmetric about $\theta.$ Then  $X$ is $\beta$-symmetric about $\theta$ for each $\beta \in (0,1].$
\label{bcsymmetry}
\end{proposition}


 Proposition \ref{x0betasim} describes the $\beta$-symmetry points of $X.$

  \begin{proposition} \label{x0betasim}
  	Let $X$ be a $\beta$-symmetric random element in  $\mathbb{E}$ and $x_0 \in \mathbb{E}$ such that $IDLD^{\beta'}(x_0,P) = \frac{1}{2}$ for every $0< \beta' \leq \beta.$ Then $x_0$ is a $\beta$-symmetry point.
  \end{proposition}

 \textbf{P. 3} establishes that  the IDLD is monotone relative to the deepest point. Several auxiliary results that appear in the Appendix must be stated before proving this property.

  \begin{P3*} \label{propP3}
  	Let  $\mathbb{E}$ be a separable Banach space and $\mathbb{E}'$ the corresponding dual space. Let  $X$ be a random  $C$-symmetric element about
  	$\theta$ with probability measure  $P.$ Let  $Q$ be a probability measure in $\mathbb{E}'$ independent of $P$ and assume that for every
  	$f \in \mathbb{E}',$  $f(X)$  has unimodal density function $g,$ symmetric about  $f(\theta)$ and fulfills
   \begin{equation} \label{desigualdadLema3Propiedad3}
 g(t) \geq 2 \frac{g(t+\lambda_{t}^{\beta})g(t-\lambda_{t}^{\beta})}{g(t+\lambda_{t}^{\beta})+g(t-\lambda_{t}^{\beta})} \ \forall t \in \mathbb{R}, \mbox{ }  Q-a.s.
 \end{equation}

    Then, for every $x\in \mathbb{E}$ and $\beta \in (0,1],$
  	$$IDLD^{\beta}(x,P) \leq IDLD^{\beta}((1-t)\theta + xt,P) \ \ \ \mbox{ for every } t \in [0,1].$$
  \end{P3*}

\begin{remark}
Inequality (\ref{desigualdadLema3Propiedad3}) holds for the standard normal distribution. Hence, the projections of a Gaussian process fulfill \textbf{P. 3.}
\end{remark}

\begin{P4*}
\label{vanishinf}

Assume that $\sup_{\|u\|=1} \left\{ f: f(u) \leq \epsilon \right\}=O(\epsilon),$ where $O(\epsilon)$ is a function such that $\lim_{\epsilon\rightarrow0}O(\epsilon)=0.$ Then,
$$ \displaystyle \lim_{||x|| \to + \infty} IDLD^{\beta}(x,P) = 0.$$
\end{P4*}

We omit the proof of \textbf{P. 4.} since it is analog to the proof of Theorem 1 given by Cuevas and Fraiman \cite{CF09}.


\begin{P5*} \label{P5} Let $X \in \mathbb{E}$ be a random continuous element and $\beta \in (0,1].$ Then
 $IDLD^{\beta}(\cdot,P): \mathbb{E} \rightarrow \mathbb{R}$ is continuous.
\end{P5*}

\begin{P6*} Le $\mathbb{E}$ be a separable Banach space and $(P_n)_{n \geq 1}$ defined on $\mathbb{E}$ such that $P_n$ weakly converges to $P$ . Then, for every $\beta \in (0,1],$  $IDLD^{\beta}(x,P_n) \rightarrow IDLD^{\beta}(x,P).$
\end{P6*}


\section{Empirical Version and Asymptotic Results}

In this section we introduce the empirical counterpart of the IDLD and give the main asymptotic results. Let $Z$ be an absolutely continuous random variable with distribution $P^1.$ Suppose given  $Z_1, \dots, Z_n$  iid random variables, also with distribution $P^1,$  denote $F$ to the cumulative function distribution of $P^1$ and $F_n$ to the empirical counterpart.

First of all, recall  the definition ofPaindaveine and Van Bever \cite{PB13}  of the empirical local unidimensional simplicial depth. Let $ELD_{S}^{\beta(k)} (\cdot,F_n) : \mathbb{R} \longrightarrow \left[ 0,1/2 \right].$ Then
\begin{equation*}
ELD_{S}^{\beta(k)} (z,F_n)  = \frac{2}{\beta(k)^2} \left[F_n\left( z + \lambda_{z,n}^{\beta(k)} \right)  - F_n(z) \right] \left[F_n(z) - F_n\left( z-\lambda_{z,n}^{\beta(k)} \right) \right],
\end{equation*}
where
$
\lambda_{z,n}^{\beta(k)} = \inf_{\lambda>0} \left\lbrace  F_n\left( z + \lambda_{z,n}^{\beta(k)} \right)  - F_n\left( z - \lambda_{z,n}^{ \beta(k)} \right)   = \beta(k) \right\rbrace 
$
 and $\beta(k)=\frac{[n\beta]}{n}$ where  $[\cdot]$ is the integer part function.
Remark \ref{propiedades chiquitas} entails the well-definedness of the empirical neighborhood width, $\lambda_{z,n}^{\beta(k)}.$
\begin{remark} \label{propiedades chiquitas}
Let $\beta \in (0,1]$ and $Z_1, \dots, Z_n$ be  a random sample of iid variables with distribution $F.$ Given $z \in \mathbb{R},$ put,
for each $1 \leq j \leq n,$
$d_j(z) = |Z_j - z|$ and let $d^{j}(z)$ denote the $j-$th order statistics of $d_1(z), \dots, d_n(z).$
Let $k = [n \beta],$  it is clear that $ \# \{ Z_j \ : Z_j \in \ [z-d^{k}(z), z+d^{k}(z)] \} =  k, $ almost surely.
Hence,
$F_n(z+d^{k}(z)) - F_n(z-d^{k}(z)) = \frac{[n \beta]}{n} = \beta(k),$ and   the  empirical neighborhood width is  $\lambda_{z,n}^{\beta} = d^{k}(z).$
\end{remark}

Then the empirical counterpart of IDLD is given as follows.

\begin{defn} Let $\beta \in (0,1],$ $X: \Omega \to \mathbb{E}$ be a continuous random element and $X_{1}, \dots, X_{n}$  a random sample with the same distribution as $X.$ Let $k=[n \beta].$ For each $x \in \mathbb{E}$ and $f \in \mathbb{E}',$ define
\begin{equation} \label{lambdaempirico}
\lambda_{f(x),n}^{\beta(k)} = inf \left\{ \lambda > 0 : F_{f,n}(f(x) + \lambda) - F_{f,n}(f(x) - \lambda) = \frac{k}{n} \right\}.
\end{equation}
Let $\beta(k) = \frac{k}{n}.$ The empirical version of IDLD of locality level  $\beta(k)$  is
\begin{equation} \label{ELIDD}
EIDLD^{\beta(k)}(x,P) = IDLD^{\beta(k)}(x,P_n).
\end{equation}
\end{defn}

%
%
%
%
%

%

The theorems below establish the uniform strong convergence of the empirical counterpart of the univariate simplicial local depth to the population counterpart. The proofs appear in the Appendix as well as several lemmas needed to establish these results. 

\begin{theorem} Let $\mathbb{E}$ be a separable Banach space with dual  space  $\mathbb{E}'.$ Suppose  given $X_1, \dots, X_n$ a random sample of elements on $\mathbb{E}$ with probability measure $P$ and $\beta \in (0,1].$
Then, we have

\begin{enumerate}[(a)]
\item \begin{equation}
E \left( \sup_{x \in \mathbb{E}} \Big| ELD_S^{\beta(k)}(f(x),P_{n,f}) - LD_S^{\beta}(f(x),P_f) \Big| \right) \xrightarrow[n \to + \infty]{} 0 \ \mbox{ for every } \ f\in \mathbb{E}'.
\end{equation}

\item \begin{equation}
E \left( \sup_{x \in \mathbb{E}} \Big| EIDLD^{\beta(k)}(x,P_n) - IDLD^{\beta}(x,P) \Big| \right) \xrightarrow[n \to + \infty]{} 0.
\end{equation}
\end{enumerate}
\end{theorem}
\begin{theorem} \label{consistenciactp} Let $X$ be a random element on $\mathbb{E}$ a separable Banach space with associated probability measure $P$ such that
$E(f(X)^2) < +\infty \ \mbox{ for every } \ f \in \mathbb{E}'.$ Let $X_1, \dots, X_n$ be a random sample following the same distribution as  $X$ and $\beta \in (0,1].$ Then,
\begin{equation*}
\mathbb{P} \left( \sup_{x \in \mathbb{E}} \Big| EIDLD^{\beta(k)}(x,P_n) - IDLD^{\beta}(x,P) \Big| \xrightarrow[n \to +\infty]{} 0 \right) = 1.
\end{equation*}
\end{theorem}


\section{Local Depth Regions}

In this section we define the \textit{$\alpha$ local depth inner region at locality level $\beta.$}
 Ideally, these central regions will be invariant of the coordinate system and nested. We also study, under mild regularity conditions, the asymptotic behavior.

Denote by $LD^{\beta}$ a local depth measure and $ELD^{\beta}$ its empirical counterpart. In particular, one can consider the integrated dual local    depth defined in Section \ref{IDLDsection}.
\begin{defn}
Let $\mathbb{E}$ be a separable Banach space, let $X: \Omega \rightarrow \mathbb{E}$ a random element with associated probability measure $P.$
Fix  $\beta \in (0,1],$ a locality level,  and $\alpha \in [0,\frac{1}{2}].$ The \textit{local  depth inner region at locality level} $\beta$
\textit{of level} $\alpha$ is defined to be
\begin{equation}
\label{ldregion}
R_{\beta}^{\alpha} = \left \{ x \in \mathbb{E}: \ LD^{\beta}(x,P) \geq \alpha \right \}.
\end{equation}
\end{defn}

Let $X_1, \dots, X_n$ be a random sample of elements on $\mathbb{E}.$ Then the empirical counterpart of $R_{\beta}^{\alpha}$ is
$$ R_{n}^{\alpha} =  R_{n,\beta}^{\alpha} = \left \{ x \in \mathbb{E}: \ ELD^{\beta}(x,P_n) \geq \alpha \right \}. $$

Throughout this section the locality level $\beta$ will remain fixed, hence we write  $R^{\alpha}$ (respectively, $R_n^{\alpha}$) for $ R_{\beta}^{\alpha}$  (respectively. $R_{n,\beta}^{\alpha}$) when no ambiguity is possible.

\begin{remark} \label{PropiedadesRegionProfundidad1}
If $\mathbb{E}$ is a finite dimensional space, and $LD$ is invariant under orthogonal transformations, then $R^{\alpha}$ inherits this property.
\end{remark}
\begin{remark} \label{PropiedadesRegionProfundidad2}
If  $\alpha_1 \leq \alpha_2,$ then  $R_{\beta}^{\alpha_2} \subset R_{\beta}^{\alpha_1}.$
 \end{remark}

Theorem \ref{consistRalfa} shows that the empirical $\alpha$ local depth inner region at locality level $\beta$ is strongly consistent with  its corresponding population counterpart, under regularity conditions. We omit the proof since it is analog to the prove given by Zuo and Serfling \cite{ZS00} in Theorem 4.1.

\begin{theorem} \label{consistRalfa}
Let $\mathbb{E}$ be a separable Banach space and let $X: \Omega \rightarrow \mathbb{E}$ be a random element with associated  probability measure
 $P.$  Assume that
\begin{enumerate}[a)]
 \item $ \displaystyle LD^{\beta}(x,P) \xrightarrow[ \| x \| \to +\infty]{} 0.$
 \item $ \displaystyle \sup_{x \in \mathbb{E}} \left| ELD^{\beta}(x,P) - LD^{\beta}(x,P) \right| \xrightarrow[n \to +\infty]{} 0$ a.s.
\end{enumerate}
Then, for every $\epsilon > 0,$ $0 < \delta < \epsilon,$ $0 < \alpha$ and sequence $\alpha_n \rightarrow \alpha$:
\begin{enumerate}[(I)]
 \item There exists  an $n_0 \in \mathbb{N}$ such that
 $R^{\alpha + \epsilon} \subset R_{n}^{\alpha_n + \delta} \subset R_{n}^{\alpha_n} \subset R_{n}^{\alpha_{n} - \delta}
 \subset R^{\alpha - \epsilon} a.s.$  
 \item If $P \left( x \in \mathbb{E}: \ LD_{\beta}(x) = \alpha \right) = 0,$ then $R_{n}^{\alpha_n} \xrightarrow[n \to +\infty]{}
 R^{\alpha}$ a.s.
\end{enumerate}
\end{theorem}


\section{Numerical considerations}

Local depth measures are depths conditional to a neighborhood of each point of the space. These measures are sensitive to the size of the neighborhood, which depends on a parameter that must be chosen.  They may be classified into two families. On one hand, there are those depending on a global parameter, which indicates the size of the neighborhood for every point $x$ of the space. Within this family we find the proposals made by Cuevas et al. \cite{CFF07}, Agostinelli and Romanazzi \cite{AR11}, Chen et al. \cite{CDPB09}, Sguera et al. \cite{SGL214} and Agostinelli \cite{A18}. In most of these papers, heuristic strategies are introduced for choosing this parameter in a data-driven way.

On the other hand, both in the proposal given by Paindaveine and Van Bever \cite{PB13}  as well as ours, the size of the neighborhood is adaptive to each point of the space $x.$  In these cases the parameter $\beta$ indicates a fixed probability that will be considered surrounding each point $x,$  then the neighborhood is determined by  $\lambda_x^{\beta}.$  We give a heuristic strategy for choosing $\beta$ in a data driven way.

Different $\beta$ values reveal different features that are more or less adequate to describe and analyze a dataset. 

It is well known that as locality becomes extreme local depth measures lose the nature of centrality, as pointed out by Paindaveine and Van Bever \cite{PB13}. Therein they present a detailed characterization of the behavior of the extreme localization for the univariate case. Hence, since our proposal is constructed departing from this expression it can neither capture local centrality for $\beta$ close to zero.

For high $\beta$ values (typically greater than 0.5), the correlation between depth and local depth is high, that is, the structure described by these two measures is similar. In the remaining range of values, local depths describe adequately datasets typically with more than one subpopulation.

We suggest that instead of choosing a single value for $\beta$ it is useful to provide clusters of $\beta$ that highlight different features of the dataset.

To exemplify, we return the case of the first derivatives of the growth dataset that we showed in the Introduction. Now analyze the correlations between local depths at different locality levels. To do this, we take an equispaced grid for $\beta$ values between $0.1$ and 1 with step $0.05.$ To compute IDLD we generate 500 random directions following a standard brownian motion as in Cuevas and Fraiman \cite{CF09}. We perform the heatmap of the local depth correlations calculated at the different locality levels, see Figure \ref{heatmapderivgrowth}. The dendrogram shows that there are four groups. We can order the clusters in ascending order following the locality levels that constitute them. The first one is a singleton conformed by the smallest $\beta$ value. In the second cluster beta ranges between $0.15$ and $0.25$; in the third cluster between $0.30$ and $0.50$ and finally, the fourth cluster has the remaining values. On the one hand,  IDLD for $\beta$ belonging to the first cluster is not informative. On the other hand, for the last cluster, the neighborhoods contain more than half of the points, hence the information given by the local and global depth is practically the same. The second cluster is the one that best manages to capture the central curves for girls and boys.  The third cluster presents an interface between the second and the fourth cluster, where the information of the two subpopulations is present although somewhat blurred. 
 
To stress this point, if we consider the locally deepest $ 5\%$ of the data, for locality values smaller than $0.35,$ the proportion of boys and girls among them is the same that in the entire data set, while for locality levels between $0.4$ and $0.6,$ typically, the girl's rate grow. Finally, for locality levels greater than $0.6,$ all the local deepest data correspond to girls. In Figure \ref{centralesgrowth} (which appears in the Introduction) we can notice that the local depth curves ($\beta=0.2$) better capture the centers for boys and girls while the ones corresponding to the global depth only identify one center and does not highlight the most salient information of the data set. If instead of retaining the $5\%$ locally deepest curves we increase this proportion, in every case for locality levels in the second cluster the central curves reveal the peaks for boys and for girls, while for locality levels greater than $0.5$ the central curves do not reveal this pattern.

\begin{figure}[!t]
	\centering
	\includegraphics[width=0.6 \textwidth]{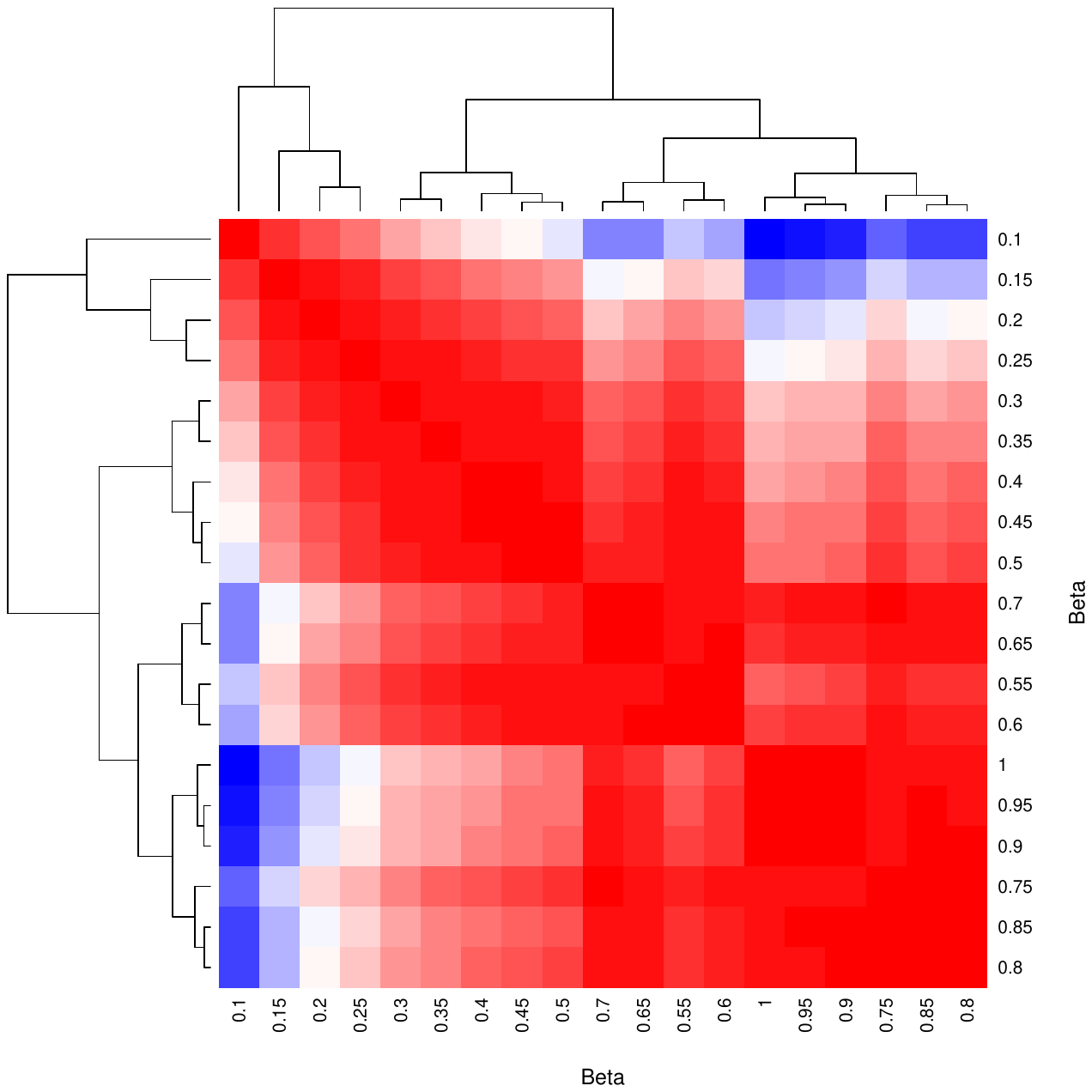}\par
	\caption{ Heatmap of the correlation matrix of the local depth for the derivative of the growth data set at different locality levels. } \label{heatmapderivgrowth}
\end{figure}

The pattern described in this example is typically observed in all the datasets that we had analyzed. The heatmap aids to analyze the matrix correlation of the local depth at different locally levels. Typically locality levels between $0.2$ and $0.4$ will capture valuable information on the data.


\section{Simulations and real data examples} \label{RDE}

 In this section, we analyze the performance of the IDLD for classification and descriptive statistics. Even though we develop the theoretical part in the most general case, for the numerical part we concentrate on functional data, specifically in the multivariate case, where there is no other local depth available. 
 
 \subsection{Descriptive Statistic. The Gross Domestic Product data set}

  In this example, we consider a two-dimensional functional data set, which is available at the International Monetary Fund website \url{https://www.imf.org/external/datamapper/NGDPD@WEO/OEMDC/ADVEC/WEOWORLD}. The database contains the annual series of Growth Domestic Product (GDP) per capita and the inflation rate (interannual variation of the consumer price index with the base year 2016), for 101 countries from 1990 to 2016. The relationship between GDP and inflation has been one of the most widely researched topics in macroeconomics, where not only the marginal behavior of each of these variables provides valuable information, but also the compound analysis between them.  Although the database has measurements since 1980, many countries did not report their inflation rates during the 1980s, hence we analyze the series beginning in 1990. The data corresponding to each of the variables are measured in different units (which differ in their order of magnitude, see Figure \ref{FMIdata2D}), therefore both series were standardized to have comparable scales by dividing the GDP by 400. The objective of this analysis is to adequately describe the data set that is being analyzed. At first glance, the univariate GDP and inflation charts do not present any clear pattern.

 \begin{figure}[!t]
 	\centering
 	\includegraphics[width=0.5 \textwidth]{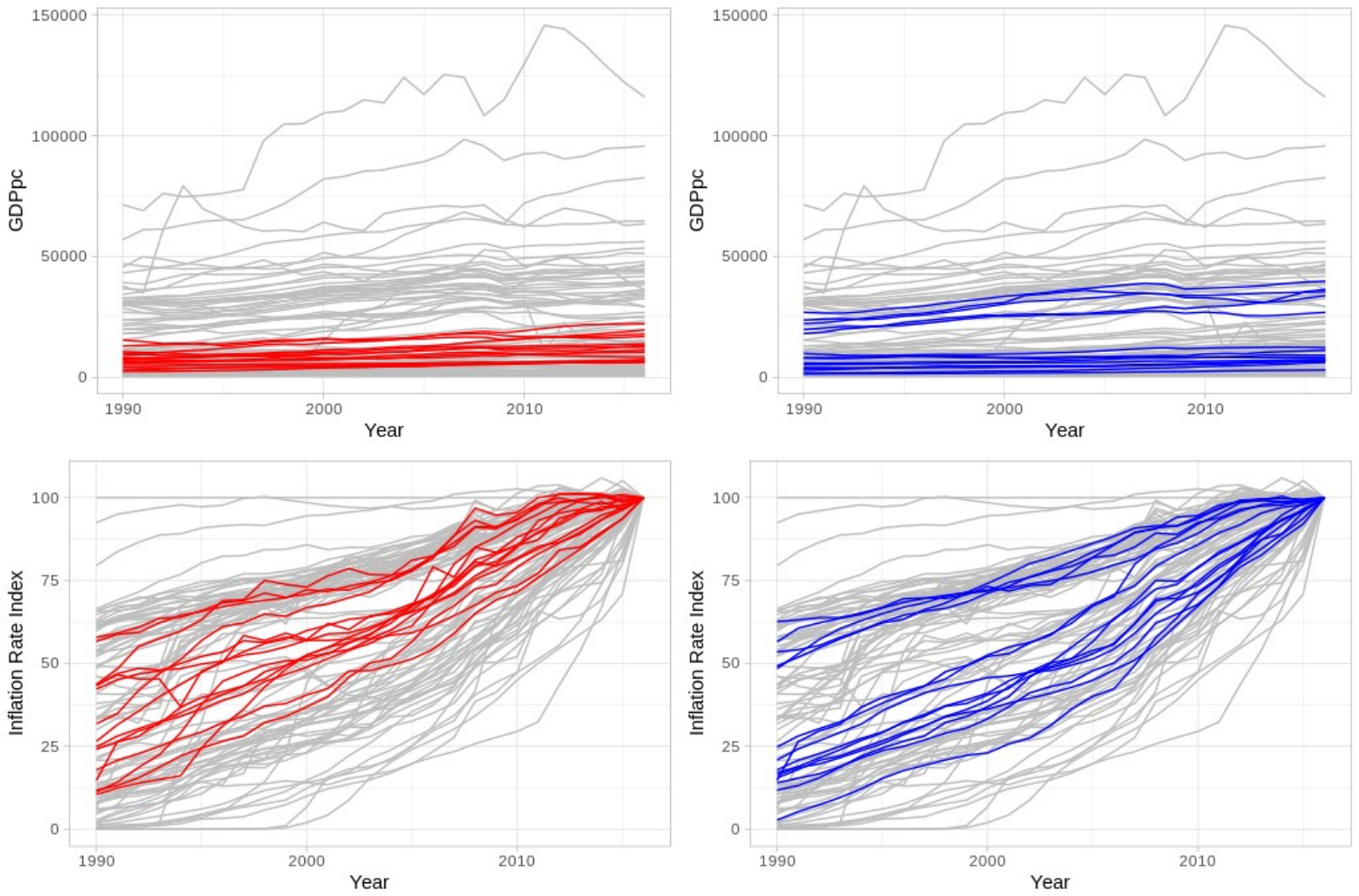}\par
 	\caption{The upper panels show the trajectories for the GDP while the lower panels exhibit the trajectories for the inflation rate. The red curves highlight the $10\%$ deepest curves given by the global depth while the blue ones highlight the $10\%$ deepest curves given by the local depth at locality level $\beta=0.4.$ } \label{FMIdata2D}
 \end{figure}

We compute  the IDLD at locality level $\beta=0.10,0.15,0.20,\dots,0.95,1.$ 
To compute 500 random directions are generated as follows. We assume that each functional coordinate of  the data belongs to $L_2(T)$ where $T$ is the period of time studied. Then, we consider the product space with the usual inner product. The trajectories are generated following a standard brownian motion distribution. 

In Figure \ref{heatmapfmi} (a) the heatmap of the correlation between the IDLD at different locality levels show a three-group structure. One group is a singleton, $\beta=0.10,$ where no structure can be appreciated since the local depths of many points are high and similar. The larger cluster contains the locality levels higher than $0.45,$ which present a high correlation with the global depth. Finally,  the remaining values constitute a group, which will be analyzed in detail.

\begin{figure}[!t]
	\centering
	\subfigure[Heatmap of the correlation matrix of the local depth for the Gross Domestic Product data set at different locality levels.]{\includegraphics[width=50mm]{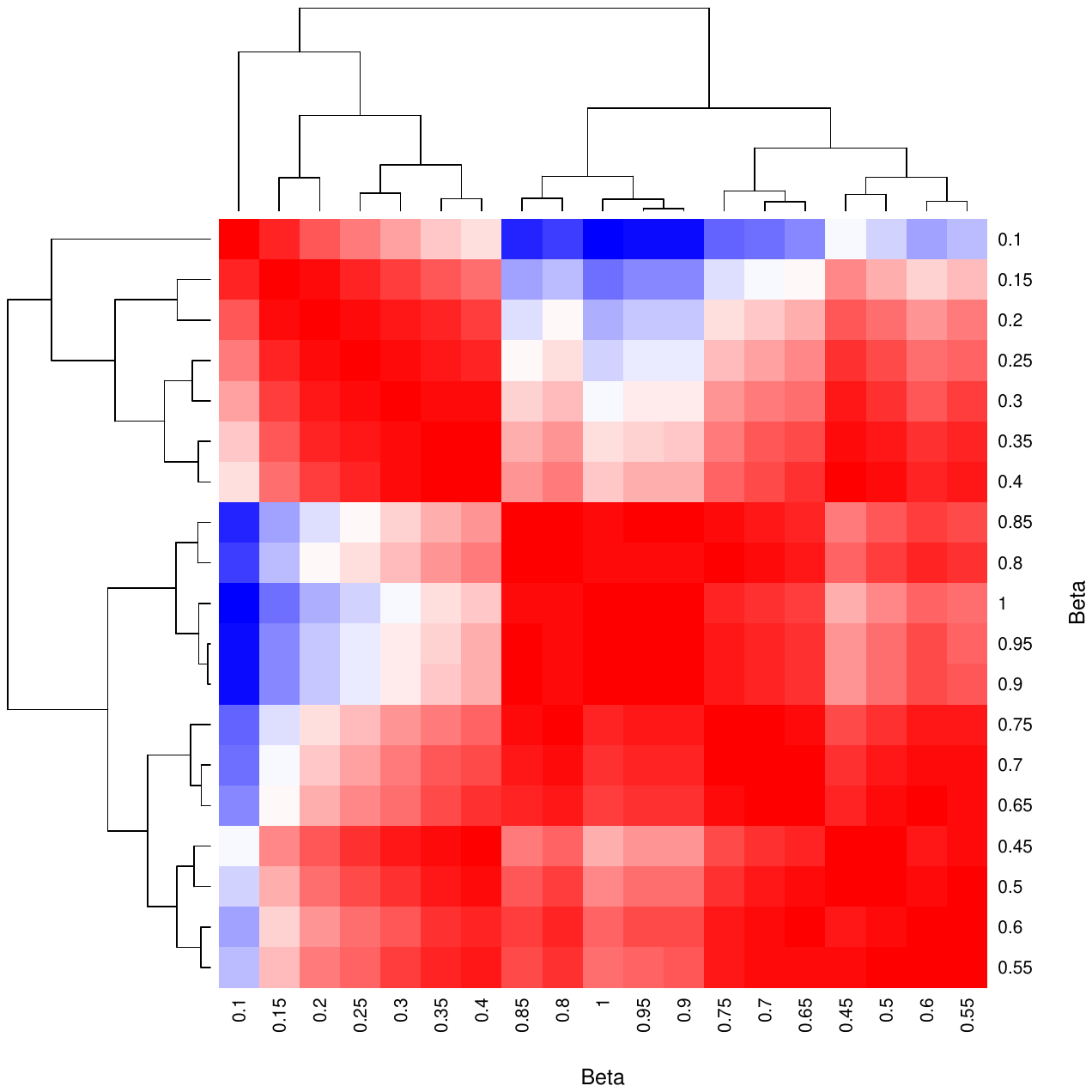}} \hspace{10mm}
	\subfigure[ Correlation matrix used for the selection of the locality parameter.]{\includegraphics[width=50mm]{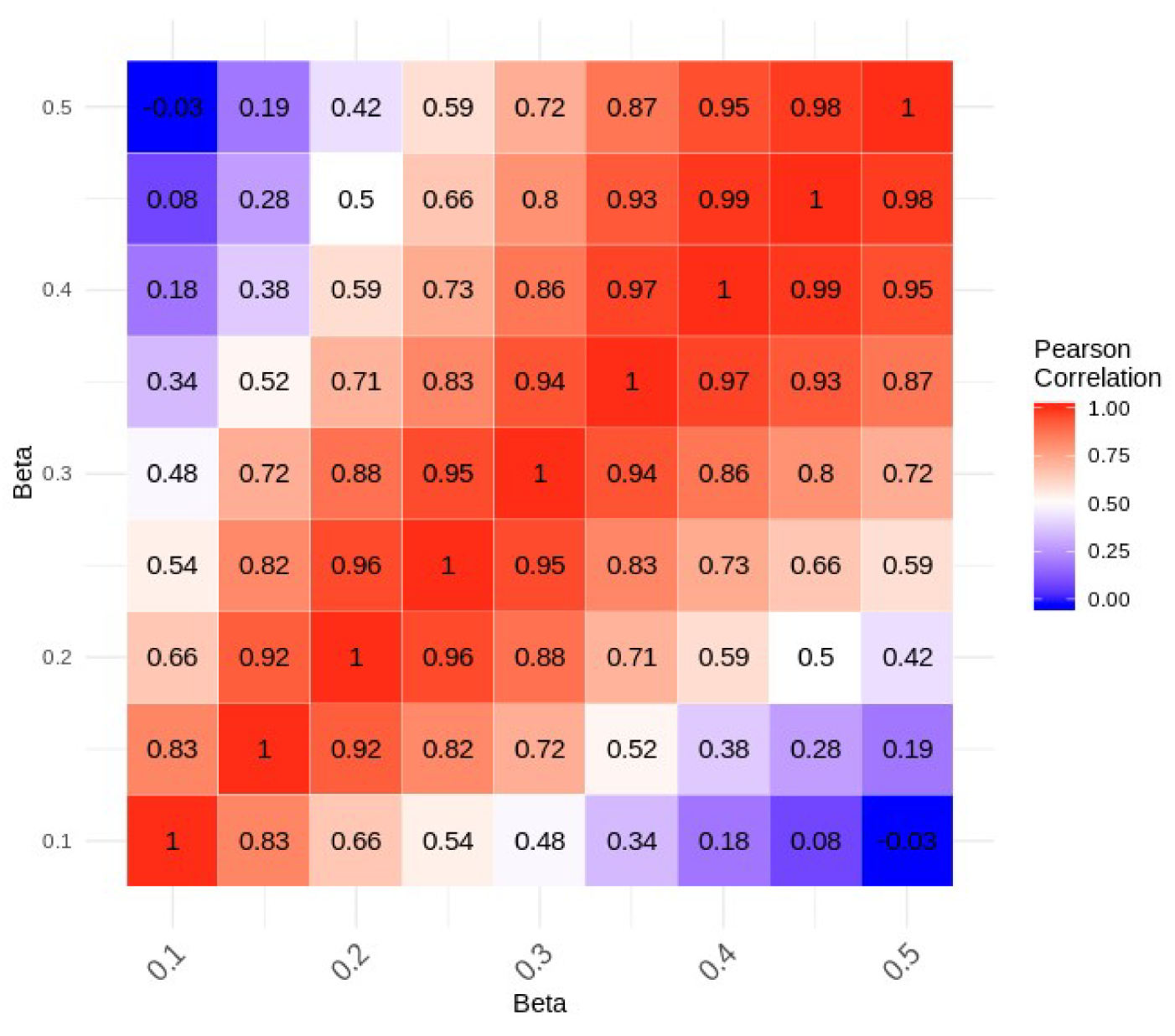}}
	\caption{ Local depth correlation analysis for FMI data.} \label{heatmapfmi}
\end{figure}

%

 Then we study the IDLD at locality level $\beta= 0.1, 0.2, \dots,0.5.$  To carry out a deeper analysis we proceed as follows. Given two locality levels $\beta_1$ and $\beta_2,$  consider the set of observations that correspond to $50\%$ of the local deepest data for at least one of the locality level, denote  $C_{\beta_1 \bigcup \beta_2 }$ to this set. Calculate the correlation matrix of the rankings of the $ILDL^{\beta_1}(C_{\beta_1 \bigcup \beta_2 }, P_n)$ and $ILDL^{\beta_2}(C_{\beta_1 \bigcup \beta_2 }, P_n),$ where $P_n$ is the distribution of the complete data set. These results appear In  Figure \ref{heatmapfmi} (b), clearly there are regions of locality level where the correlation is high, the area that stands out the most is $\beta \in [0.3, 0.5].$ The fact that the correlation, in a range of $\beta's,$ is high and stable suggests that a clear structure is being captured, otherwise, the spurious structures would be described. For this reason we chose $\beta = 0.4.$

  \begin{figure}[!t]
 	\centering
 	\includegraphics[width=0.6 \textwidth]{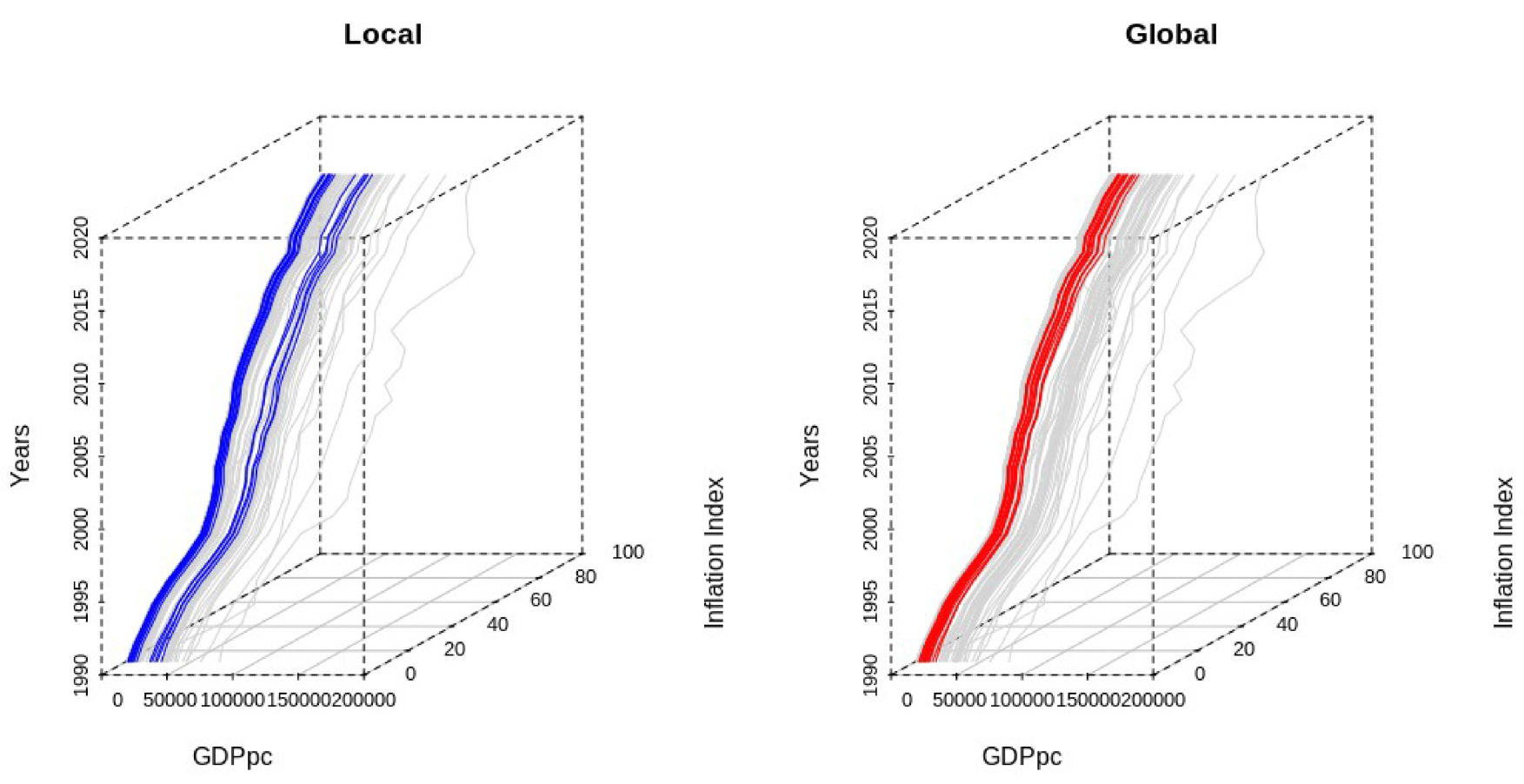}\par
 	\caption{Three dimensional representation of the trajectories of GDP and inflation, the  red curves represent the $10\%$ global deepest observations while the blue ones stand for the $10\%$ locally deepest observations at locality level $\beta=0.4$.} \label{FMI3D}
 \end{figure}

 To deepen this exercise we will analyze the $10\%$ of the deepest data, locally and globally. In Figures \ref{FMIdata2D} and \ref{FMI3D} it can be appreciated that not only the deepest observations are different but also in the local case a bimodal structure is clearly described, this pattern is described by both variables. Moreover, if we analyze the first two principal components, clearly one mode is dominated by emergent countries (for instance Paraguay and Bolivia) while the other mode is dominated by developed countries (for example UK and Spain). On the other hand, among the deepest curves for the global depth, only  emergent countries appear (see Figure \ref{pcaFMI}).

 \begin{figure}[!t]
 \centering
  \includegraphics[width=0.6 \textwidth]{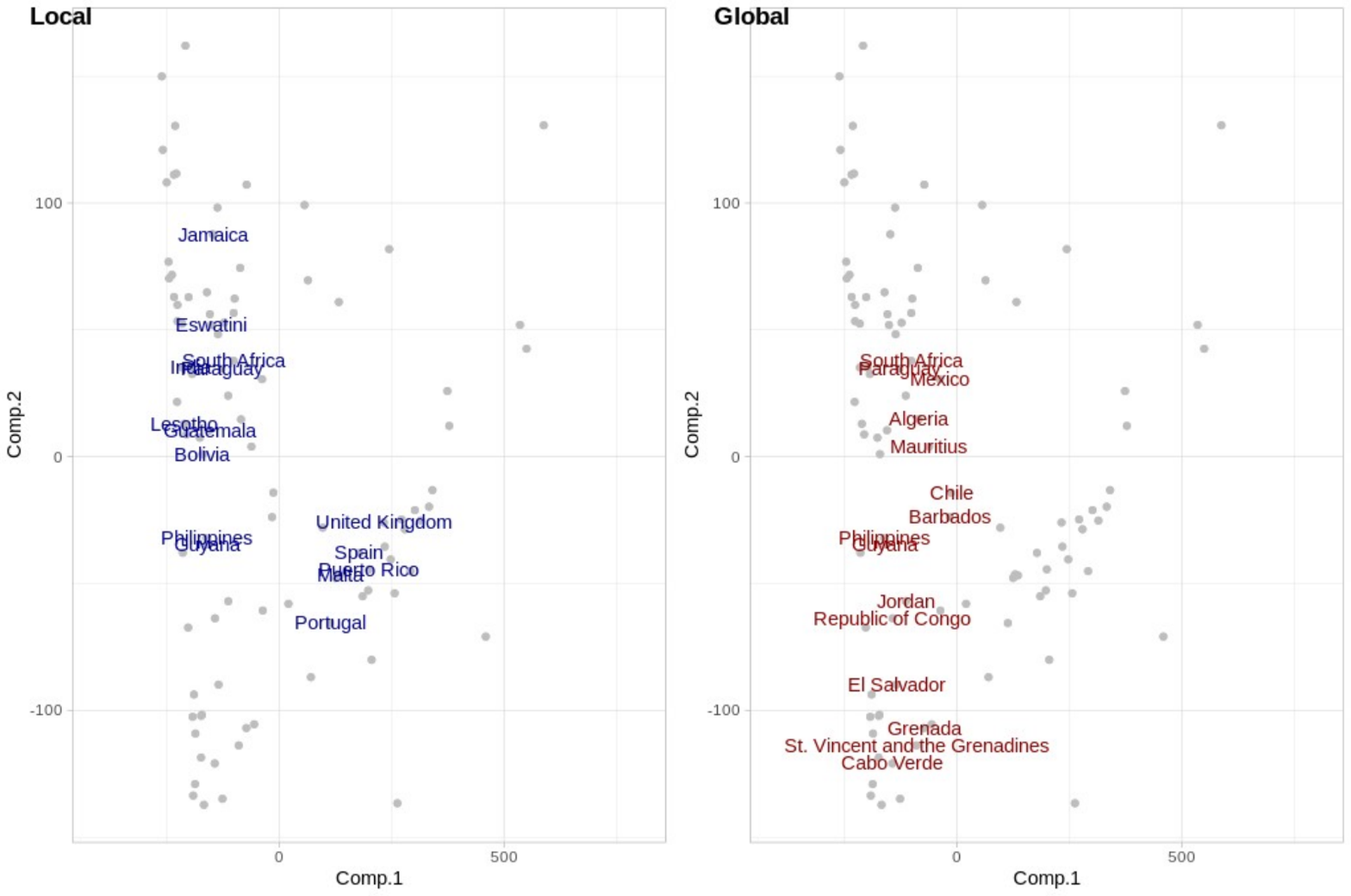}\par
 \caption{Scatter plot of the two first principal components of the data sets, the red dots represent the $10\%$ global deepest observations while the blue ones stand for the $10\%$ locally deepest observations at locality level $\beta=0.4$.} \label{pcaFMI}
 \end{figure}

\subsection{Classification}
 
 In this section we perform a simulation study as well as analyze two well-known data sets and show that the IDLD improves the max-depth classification. 

\subsubsection{Simulation}

Five models are studied for classification purpose. 
In every case we consider bivariate functional data. In all the cases, each functional  curves are obtained from the process $m_{\cdot jk}(t)=m_{jk}(t)+\epsilon_{\cdot jk}(t),$ where $m_{jk}$ is the mean function of the class $j=1,2$ and for coordinate $k=1,2$ and $\epsilon_{\cdot jk}$ is a gaussian process with zero mean and covariance operator $cov\left( \epsilon_{\cdot j}(s),\epsilon_{\cdot j}(t)\right)=\theta_j \exp\left(\left| s-t\right| /0.3 \right),$ with $\theta_2=\theta_1/2.$  The functions are generated in the interval $\left[ 0,1\right]$ on a grid of $71$ equidistant points. The models are inspired in those of Cuesta-Albertos et al. \cite{CFD17}.

\begin{itemize}
	\item \textit{Model 1:} The mean of Population 1 is $m_{1}(t)=\left(m_{11}(t),m_{12}(t)\right) =\left( 30(1-t)t^{1.1}, 3 \sin(2\pi t)\right)$ and for Population 2 $m_{2}(t)=\left(m_{21}(t),m_{22}(t)\right) =\left( 30(1-t)^{1.1}t, \sin(2\pi t)\right).$ Each coordinate of Population 1 has error process by $\epsilon_{\cdot 1k}$ and for Population 2 $\epsilon_{\cdot k},$  where  $\theta_1=2.$    The sample size of the training and validation samples for each population is 100.
	\item \textit{Model 2:} Population 1 and 2 follow the same distribution as in \textit{Model 1}, but the sample size of the training and validation sets for Population 1 is 70, while for Population 2 is 100 observations.
	\item \textit{Model 3:} The mean of Population 1 is the same as in \textit{Model 1}, $m_1(t),$ while  Population 2  is a mixture of two subgroups one with mean  $m_{20}(t)=\left( 35 (1-t)^{1.1}t, \sin(2\pi t)\right)$ and the other one with mean $m_{20}(t)=\left( 25 (1-t)^{1.1}t, 5 \sin(2\pi t)\right),$ the mixing proportion is $0.5.$ The error process for  both coordinates of Population 1 is generated with $\theta_1=0.5,$ and both coordinates of Population 2 with $\theta2.$  The sample sizes for the training and validation samples for each population is 50 observations.
	\item \textit{Model 4:} Population 1 is a mixture of two groups with means $m10(t) = (22(1-t)(t^k), 7 \sin(2\pi t))$ and $m11(t) = (30(1-t)(t^k), 3 \sin(2 \pi t)).$ Population 2 is a mixture of two groups with means $m20(t) = (34((1-t)^k)t, sin(2 \pi t))$ and $m21(t) = cbind(26((1-t)^k)t, 5 \sin(2 \pi t)).$ In both cases the mixing proportions is $0.5.$ The error process for both coordinates of Population 1 is generated with $\theta_1=0.5,$ and both coordinates of Population 2 with $\theta2.$  The sample sizes for the training and validation samples for each population is 50 observations.
	\item \textit{Model 5:} Has the same distribution as \textit{Model 4}, but in this case each subgroup is a group.  
\end{itemize}

The simulation results are based on 200 independent replicates. The observations are assigned to the 
Population with highest local depth  for fix  $\beta=0.05,0.1,0.2,\dots,0.9,1$ and also for a 5-fold cross-validated  $\beta.$  For each of them we analyze missclasification error rate at different locality levels.

\begin{table}[ht]
	\centering
	\begin{tabular}{r|rrrrrrrrrrr|r}
		\hline
		 & \multicolumn{12}{c}{$\beta$}   \\ 
		\hline
		 & 0.05 & 0.1 & 0.2 & 0.3 & 0.4 & 0.5 & 0.6 & 0.7 & 0.8 & 0.9 & 1 & cv \\ 
		\hline
		M1 & 0.21 & \textbf{0.16} & 0.18 & 0.19 & 0.21 & 0.22 & 0.23 & 0.24 & 0.24 & 0.25 & 0.26 & 0.17 \\ 
		M2 & 0.42 & 0.23 & \textbf{0.21} & 0.25 & 0.29 & 0.31 & 0.33 & 0.34 & 0.35 & 0.36 & 0.37 & 0.23 \\ 
		M3 & 0.35 & \textbf{0.22} & \textbf{0.22} & 0.26 & 0.30 & 0.35 & 0.39 & 0.42 & 0.43 & 0.44 & 0.44 & \textbf{0.22} \\ 
		M4 & 0.33 & \textbf{0.25} & 0.27 & 0.29 & 0.31 & 0.33 & 0.34 & 0.35 & 0.35 & 0.35 & 0.35 & 0.26 \\ 
		M5 & 0.25 & \textbf{0.11} & 0.14 & 0.16 & 0.19 & 0.20 & 0.22 & 0.23 & 0.25 & 0.26 & 0.26 & 0.12 \\ 
		\hline
	\end{tabular}
    \caption{Mean misclassification error rate for Models 1-5  for fixed $\beta$ and cross validated $\beta$. The best result appears in boldface.}
    \label{tablasimul}
\end{table}


\begin{figure}[!t]
	\centering
	\includegraphics[width=0.7 \textwidth]{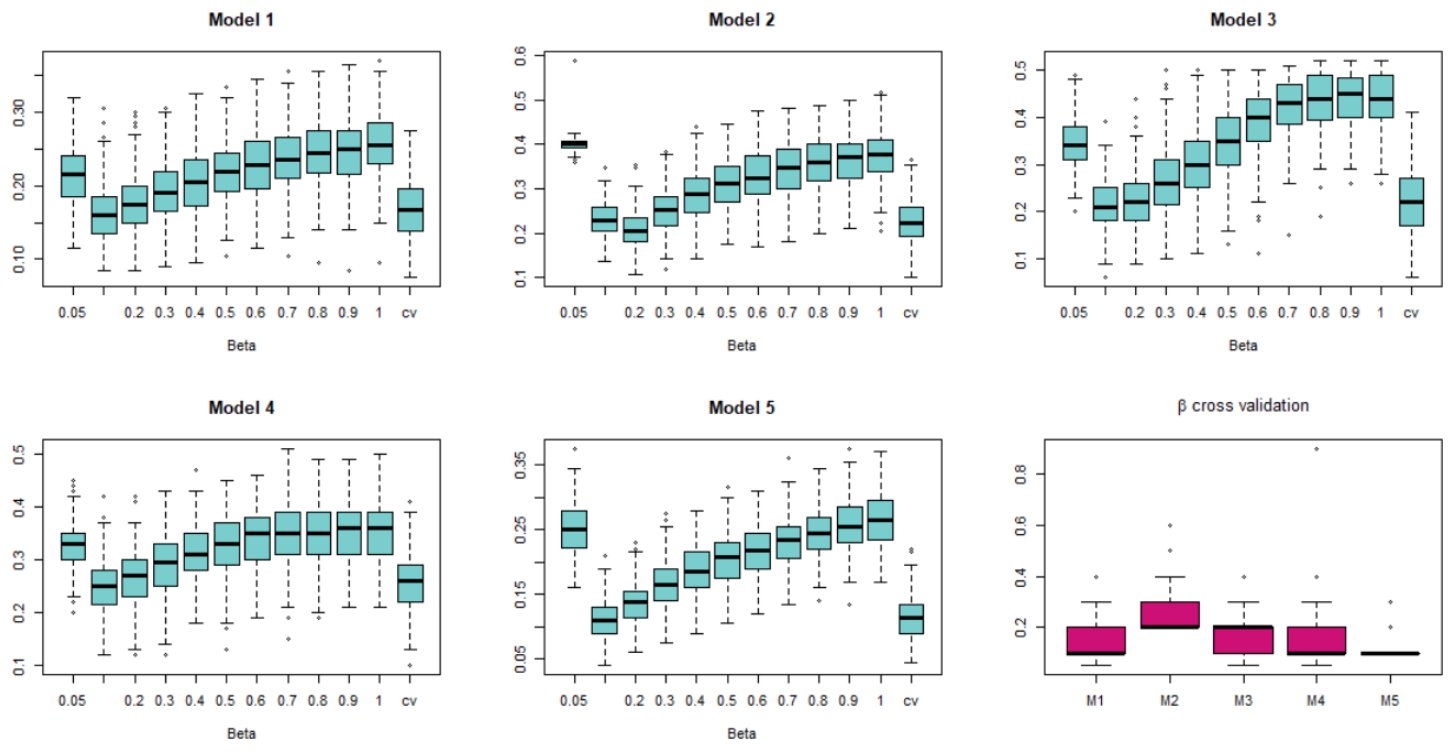}\par
	\caption{Boxplots of the misclassification rate for Models 1 to 5 (upper panels and lower  right-hand side and central panels ). Boxplots of $\beta$ selected by cross validation in Models 1 to 5 (lower  felt-hand side panel).}
	\label{bpsimul}
\end{figure}

In Table \ref{tablasimul} the mean misclassification error for IDLD at different locality levels is presented. The benchmark is given when $\beta=1,$ which is the IDD, in every case these values are among the highest. The last column presents the misclassification rate for $\beta$ obtained via cross-validation, these results are very close to the best ones. In every case, for the smallest locality level considered the IDLD is not successful, since it cannot identify the central observations of the sample. In all the cases the best performance is attained either for $\beta=0.1$ or $0.2.$ Typically the misclassification level increases as $\beta$ increases. The misclassification rate for Model 1 is smaller than for Model 2.  For Models 3 and 4, where at least one of the groups is a mixture of two distributions, for $\beta$ greater than $0.6$ the results are practically the same because the local features are not captured. In Model 5, for $\beta=0.1$ and $0.2$ it correctly identifies the four groups. These features are also captured by the boxplots in Figure \ref{bpsimul}.
The performance for $\beta$ selected by cross-validation is similar to the performance given in the best case for the fixed locality parameter. In almost every case $\beta$ is smaller than $0.3$ and typically is either $0.1$ or $0.2,$  these results are consistent with the ones given for fixed locality level.

 The R code is available at \url{https://github.com/lfernandezpiana/lidDepth/tree/master/Simulations}.

 \subsubsection{The Growth data set}

We return to the  Growth dataset, now we will analyze it from the classification perspective. 
To begin with, our aim is to classify the growth curves for a sample of boys and girls, we start working with the trajectories. Since there is no learning and test sample we shall proceed following the ideas from Sguera et al. \cite{SGL214}. A thousand replications are made with the following scheme. On each replicate the training sample is randomly chosen, respecting the proportions of girls and boys in the data set, 40 curves correspond to girls and 30 to boys. The remaining observations are classified, using the criterion of maximum local depth, taking into account different locality levels, $\beta = 0.1, 0.2, \dots, 0.9, 1,$ where $\beta = 1$ is the global depth.
It is important to highlight that we consider the same locality level for boys and girls for the sake of simplicity since the IDLD is stable in certain  $\beta$ intervals.

The left-hand side panel in Figure \ref{misclassgrowth} presents the boxplot of the misclassification rate, for each locality parameter $\beta.$ Clearly, the misclassification depends on the locality parameter, and in this case, the optimal locality level is $\beta = 0.2,$ where with median (respectively mean)  misclassification rate is $8.7\%$ (respectively $10.78\%$). If instead of the IDLD at locality level $0.2$ we had considered the global depth measure the median (respectively mean) misclassification rate is   $18.06\%$ (respectively $17.39\%,$) hence the relative improvement is $38\%$ (respectively $52\%$). It is important to highlight that these results are similar to those obtained by Sguera et al. \cite{SGL214}.

\begin{figure}[!t]
\centering
 \includegraphics[width=0.6 \textwidth]{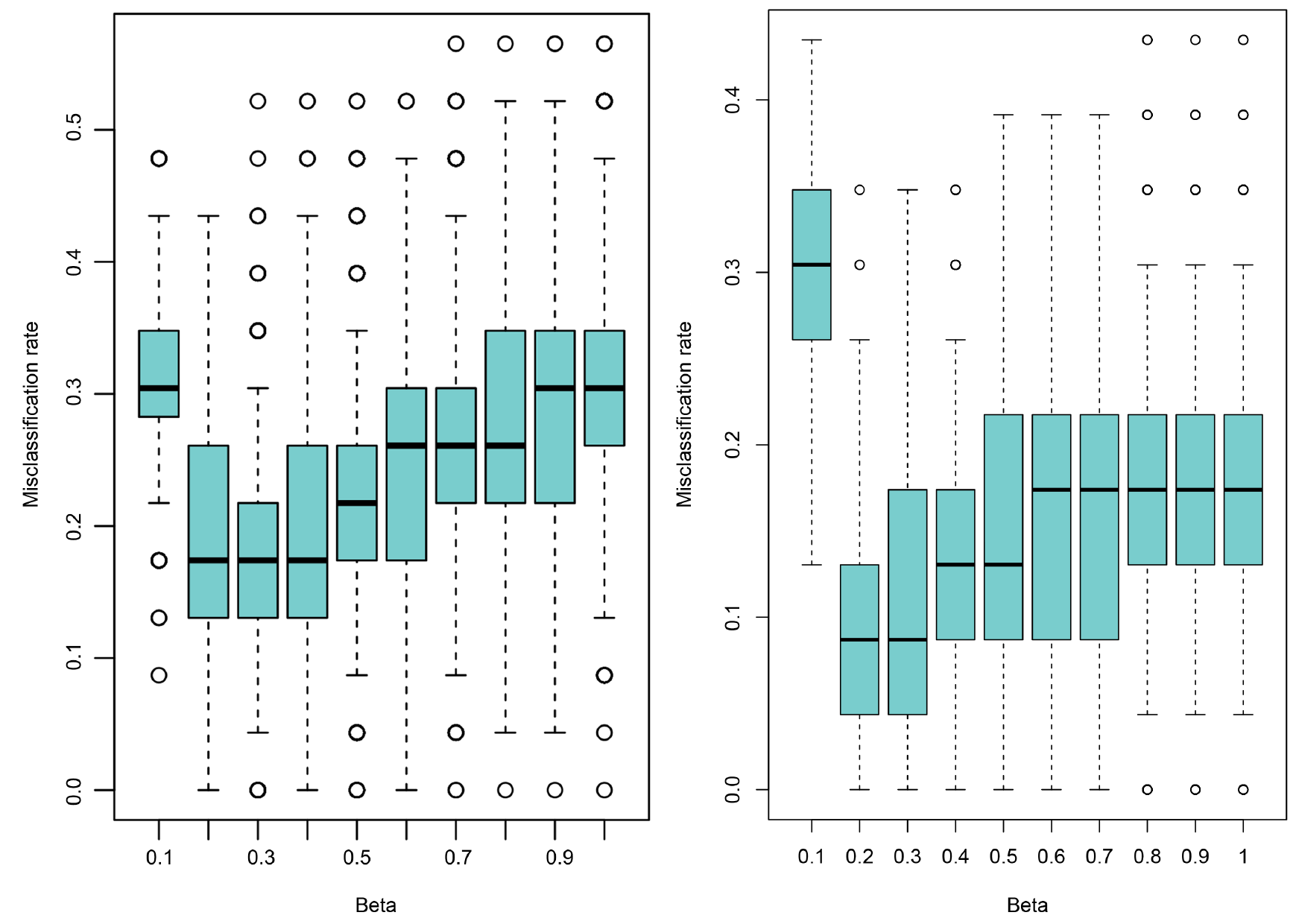}\par
\caption{Boxplots for the misclassification rates at different locality levels for the trajectories (left hand side panel) and for the derivatives (right hand side panel) of the growth dataset.} \label{misclassgrowth}
\end{figure}

Finally, we repeat the exercise with the derivatives of the growth curves. The right-hand side in Figure \ref{misclassgrowth} shows the boxplots corresponding to the misclassification rate for the velocity of the growth curves are shown. Once again, we can conclude that the classification is better when considering local depths with locality levels between $0.2$ and $0.4.$ Overall, better classification rates are achieved considering the trajectories instead of the derivatives.  


\subsubsection{The Character Trajectory data set}

We analyze the bivariate functional writing data which is a part of the Character Trajectory data set that is available at the UCI machine learning repository (Bache and Lichman \cite{BL13}). The data measure the horizontal and vertical coordinates of a pen tip while a person writes either the letter `s',   `u' or  `v', there are $n = 133$ observations for the letter `s' and $n=131$ observations for the letter `u' and $n=155$ for the letter `v'. The observed data were first interpolated to obtain $T = 100$ equally spaced time points. The trajectories for the letter `s', `u' and `v' are plot in Figure \ref{letrasuv} (a),  (b) and (c) respectively.
\begin{figure}[!t]
	\centering
	\includegraphics[width=0.6 \textwidth]{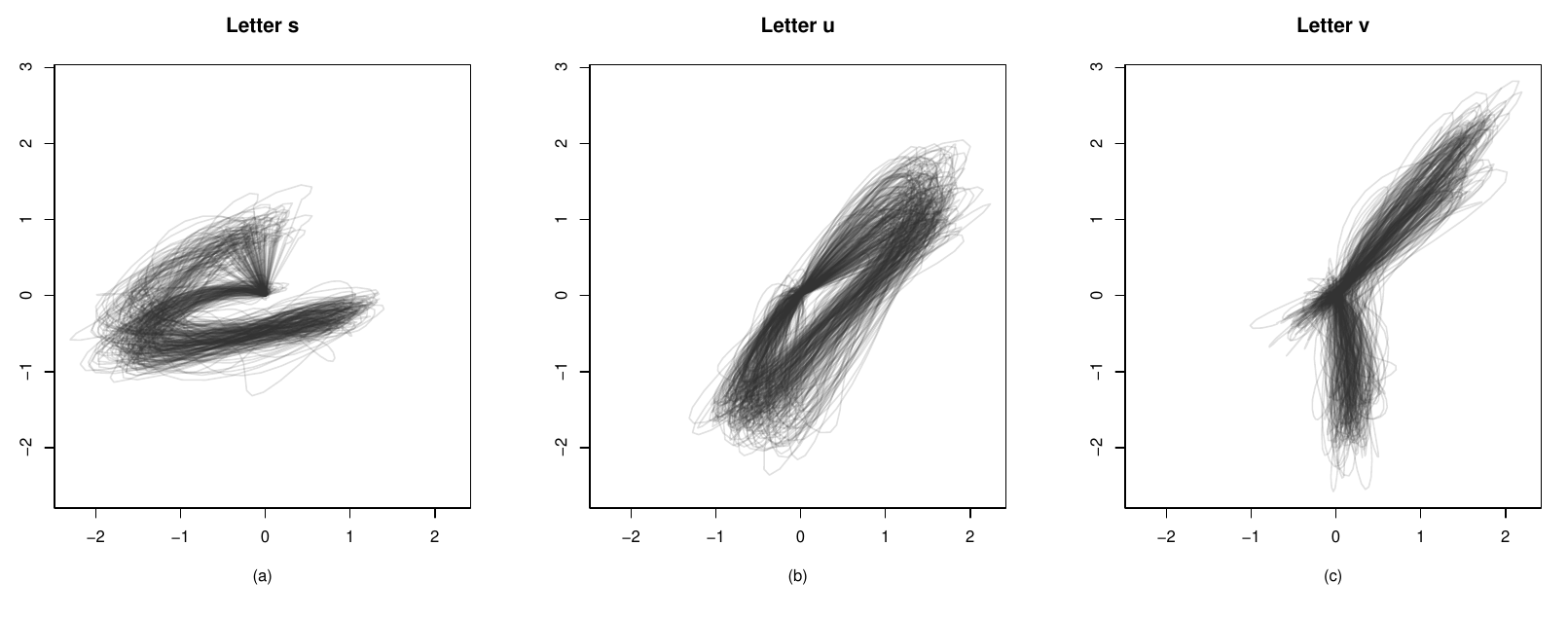}\par
	\caption{ $X$ and $Y$ position for letters `s', `u' and `v'.} \label{letrasuv}
\end{figure}

Our aim is to classify the letter trajectories in the correct category. Since there is no learning and test sample we shall proceed as in the previous example. Five thousand replications are made with the following scheme. On each replicate, the training sample is selected following a probability proportional to the size of each letter sampling scheme. The remaining observations are classified, using the criterion of maximum local depth, taking into account different locality levels, $\beta = 0.05,0.10, 0.15, \dots, 0.95, 1,$ where $\beta =1$ is the global depth.


For $\beta = 0.05$ the misclassification rate is much higher than for any other $\beta$ value. The lower misclassification rates are attained for $\beta$  between $0.1$ and $0.2,$ which remain moderate for $\beta$ between $0.25$ and $0.45.$ Finally, for locality levels above $0.5,$ the malclassification rates are higher. The lowest mean misclassification rate is $3.5\%$ for $\beta=0.15,$ and the median is $3.8\%.,$ (see Figure \ref{misclassuvs})  
\begin{figure}[!t]
\centering
\includegraphics[width=0.6 \textwidth]{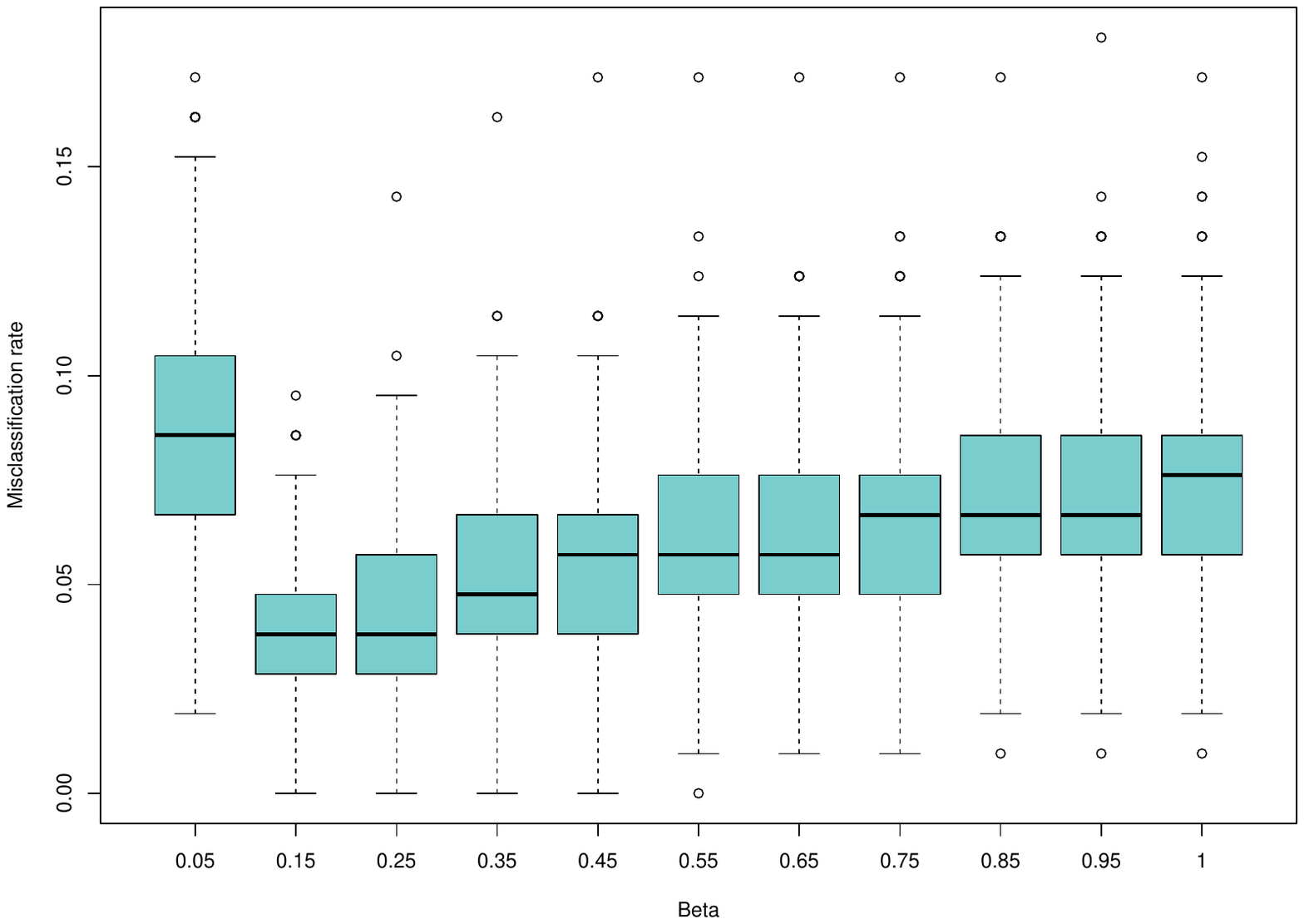}\par
\caption{ The boxplots present the misclassification rate at different locality levels for letters `s', `u' and `v'.} \label{misclassuvs}
\end{figure}

\section{Concluding Remarks}

In this paper, we introduced a local depth measure, IDLD, suitable for data in a general Banach space with low computational burden. It is an exploratory data analysis tool, which can be used in any statistical procedure that seeks to study local phenomena.
From the theoretical perspective, local depths are expected to be generalizations of a global depth measure. Our proposal has this property. Additionally, they are expected to inherit good properties from global depths, this point has been deeply analyzed in this work. Strong consistency results for the local depth and local depth regions have been proved.

From a practical point of view, we explored the use of local depth measures as a powerful tool for describing datasets, specifically if current multimodal distributions and also their use in classification problems show better performance than the use of global depths. 

A natural extension of this work is to explore classification strategies similar to $DD^G$ classifiers based on local depth similar to those proposed in Cuesta-Albertos et al. \cite{CFD17}, but it goes beyond the scope of this paper.

\section{Appendix}
\subsection{Proofs of properties \textbf{P.1-6. }}
\setcounter{equation}{0}

\par

\begin{proof}[ Proof: \textbf{P.1.} affine - invariance]
	Let $A$ be an orthogonal transformation and $Q_1$ denotes the probability measure of $h=f \circ A.$
	\begin{equation*}
	IDLD^{\beta}(Ax,P_{AX}) = \int LD_{S}^{\beta}(f(Ax),P_{f(AX)}) Q(f)  = \int LD_{S}^{\beta}(h(x),P_{h(X)}) Q_1(f).
	\end{equation*}
	Since $Q$ is the uniform measure on the unit sphere then $Q=Q_1,$ which concludes what we wanted to prove.
	
\end{proof}

\begin{proof}[Proof: \textbf{P.2.} maximality at the center]
	It is enough to show that for each $f \in \mathbb{E}'$ and $\beta' \in (0,\beta]$
	\begin{equation*}
	LD_S^{\beta}(f(\theta),P_f)  = \frac{2}{\beta'^2}
	\left[ F_f(f( \theta)) - F_f\left(  f(\theta) - \lambda_{f(\theta)}^{\beta'}\right)  \right] \left[ F_f\left(  f(\theta) + \lambda_{f(\theta)}^{\beta'}\right)  - F_f(f( \theta)) \right] = \frac{1}{2}.
	\end{equation*}
	Since the bound, $\frac{1}{2}$ is attained,  we have that $LD_S^{\beta}(f(\theta),P_f)$ has a global maximum at $\theta.$ Then, is clear that $IDLD^{\beta}(\theta,P),$ also is a global maximum.
\end{proof}

\begin{proof}[Proof: Proposition 1]
	Let $\beta \in (0,1],$ $X$ is $C$-symmetric about $\theta$ if for every  not null $f \in \mathbb{E}',$ $f(X)$ is symmetric about $f(\theta).$
	Then, for every $0<\beta' \leq \beta,$
	
	\begin{equation*}
	\beta'  = F_f \left(  f(\theta) + \lambda^{\beta'}_{f({\theta})}\right)  - F_f \left(  f(\theta) - \lambda^{\beta'}_{f({\theta})}\right)  = 2 \left( F_f \left(  f(\theta) + \lambda^{\beta'}_{f({\theta})}\right)  - F_f \left(  f(\theta) \right)  \right).
	\end{equation*}

	Finally,
	\begin{equation*}
	\frac{\beta'}{2} = F_f \left(  f(\theta) + \lambda^{\beta'}_{f({\theta})}\right)  - F_f \left(  f(\theta) \right) ,
	\end{equation*}
	which is what we wanted to show.
\end{proof}

\begin{proof}[Proof: Proposition 2]
	First note that, given  $f \in \mathbb{E}',$  $x \in \mathbb{E}$ and $\beta \in (0,1].$
	\begin{gather*}
	\beta  = F_f \left( f(x) + \lambda^{\beta}_{f({x})} \right) - F_f \left( f(x) - \lambda^{\beta}_{f({x})} \right) \\
	F_f \left( f(x) + \lambda^{\beta}_{f({x})} \right) - F_f ( f(x) )  = \beta - \left( F_f ( f(x) ) - F_f \left( f(x) - \lambda^{\beta}_{f({x})} \right) \right).
	\end{gather*}
	From the definition of $LD^{\beta}_S(x, P_1)$  is clear that,
	\begin{equation*}
	LD_{S}^{\beta}(f(x),P_f) = \frac{2}{\beta^2} \left[\beta - \left( F_f ( f(x) ) - F_f \left( f(x) - \lambda^{\beta}_{f({x})} \right) \right) \right]
	\left[ F_f \left( f(x) + \lambda^{\beta}_{f({x})} \right) - F_f ( f(x) ) \right].
	\end{equation*}
	Let $h: [0,\beta] \rightarrow \mathbb{R},$  $h(t)\ = \frac{2}{\beta^2} (\beta - t)t,$ attains a global maximum at
	$t = \frac{\beta}{2},$ hence $LD_{S}^{\beta}$  attains its maximum when
	$F_f \left( f(x) + \lambda^{\beta}_{f({x})} \right) - F_f ( f(x) ) = \frac{\beta}{2}.$ If this property is satisfied for every $0< \beta' \leq \beta,$ then, $x$ is a
	$\beta$-symmetry  point of $f(X).$
	
	Then, let $0< \beta' \leq \beta,$
	
	\begin{equation*}
	\frac{1}{2} =IDLD^{\beta'}(x_0,P) = \int LD_{S}^{\beta'}(f(x_0),P_f) dQ(f) \Rightarrow 0 = \int \left(\frac{1}{2} - LD_{S}^{\beta'}(f(x_0),P_f) \right) dQ(f)
	\end{equation*}
	For every $f \in  \mathbb{E}',$  $LD_{S}^{\beta}(f(x_0),P_f)$ is bounded by $\frac{1}{2}.$ Hence,
	$LD_{S}^{\beta}(f(x_0),P_f) = \frac{1}{2}$ $Q-a.s.$  From the first part of the proof we know that
	$f(x_0)$ is a  $\beta$-symmetry point of $f(X),$ hence $x_0$ is a  $\beta$-symmetry point of $X.$
\end{proof}

We now focus on the proof of \textbf{P.3} that  establishes that the local depth is monotone relative to the deepest point. We first show that this results holds if the distribution is unimodal. We begin by proving several auxiliary results that we will need to prove \textbf{P. 3.}

\begin{lemma}[\textbf{P.3.}] \label{lema1Propiedad3}
	Let $Z$  be an absolutely continuous, symmetric and unimodal about $t=0$ random variable with cumulative distribution function  $F.$
	Let $\beta \in (0,1],$ define the functions,
	\begin{itemize}
		\item $U(t) = F(t + \lambda_{t}^{\beta}) - F(t).$
		\item $V(t) = F(t) - F(t - \lambda_{t}^{\beta}).$
	\end{itemize}
	Then, for every  $t \in \mathbb{R}$ we have that:
	\begin{enumerate}[a)]
		\item If $t \geq 0$ $\Rightarrow$ $U(t) \leq \beta/2 \leq V(t).$
		\item If $t \leq 0$ $\Rightarrow$ $V(t) \leq \beta/2 \leq U(t).$
	\end{enumerate}
\end{lemma}

\begin{proof}
	a) It is clear that if $t=0$ then by symmetry the equality is attained.
	
	Let $t>0$ and $g_Z$ be the density function of $Z.$
	There are two possible cases to analyze:
	
	i) If $-t < t-\lambda_{t}^{\beta}:$
	
	Since $g_Z(s)$ decreases on $(0,+\infty),$ we have that
	$$ \min_{0 \leq s \leq t} g_Z(s) \geq \max_{t \leq s \leq t + \lambda_{t}^{\beta}} g_Z(s).$$
	
	By symmetry  $g_Z(-s) = g_Z(s)$ for every  $s \in [0,t],$ then
	$$\min_{-t \leq s \leq t} g_Z(s) = \min_{0 \leq s \leq t} g_Z(s).$$
	
	On the other hand, $[t-\lambda_{t}^{\beta},t] \subset [-t,t],$ since $-t<t-\lambda_{t}^{\beta},$ which implies that  $$  \min_{-t \leq s \leq t} g_Z(s) \leq \min_{t-\lambda_{t}^{\beta} \leq s \leq t} g_Z(s).$$
	
	Thus,
	
	\begin{align*}
	U(t) - V(t) & = \int_{t}^{t+\lambda_{t}^{\beta}} g_Z(s)ds - \int_{t-\lambda_{t}^{\beta}}^{t} g_Z(s) ds \leq
	\int_{t}^{t+\lambda_{t}^{\beta}} \max_{t \leq s \leq t + \lambda_{t}^{\beta}} g_Z(s) - \int_{t-\lambda_{t}^{\beta}}^{t} \min_{t-\lambda_{t}^{\beta} \leq s \leq t} g_Z(s) = \\
	& = \lambda_{t}^{\beta} \max_{t \leq s \leq t + \lambda_{t}^{\beta}} g_Z(s) - \lambda_{t}^{\beta} \min_{t-\lambda_{t}^{\beta} \leq s \leq t} g_Z(s) = \\
	& = \lambda_{t}^{\beta} \left (\max_{t \leq s \leq t + \lambda_{t}^{\beta}} g_Z(s) - \min_{t-\lambda_{t}^{\beta} \leq s \leq t} g_Z(s) \right) \leq 0
	\end{align*}
	
	ii) If, $t-\lambda_{t}^{\beta} < -t:$
	
	Observe that if $s \in [t,-t+\lambda_{t}^{\beta}],$ we have that $-s \in [t-\lambda_{t}^{\beta},-t]$ and since the density function is symmetric about $t=0,$ we know that $g_Z(s) = g_Z(-s).$ Hence,
	$$\int_{t}^{-t+\lambda_{t}^{\beta}} g_Z(s) ds = \int_{t-\lambda_{t}^{\beta}}^{-t} g_Z(s)ds.$$
	
	Since $g_Z$  decreasing, we obtain the following inequalities,
	$$ \max_{\lambda_{t}^{\beta} \leq s \leq t + \lambda_{t}^{\beta}} g_Z(s) \leq \max_{-t+\lambda_{t}^{\beta} \leq s \leq \lambda_{t}^{\beta}} g_Z(s) \leq \min_{0 \leq s \leq t} g_Z(s). $$
	
	Then,
	\begin{align*}
	U(t) - V(t) & = \int_{t}^{t+\lambda_{t}^{\beta}} g_Z(s)ds - \int_{t-\lambda_{t}^{\beta}}^{t} g_Z(s) ds = \\
	& =\int_{t}^{-t + \lambda_{t}^{\beta}} g_Z(s)ds + \int_{-t + \lambda_{t}^{\beta}}^{t+\lambda_{t}^{\beta}} g_Z(s)ds - \int_{t - \lambda_{t}^{\beta}}^{-t} g_Z(s)ds -
	\int_{-t}^{t} g_Z(s) ds = \\
	& = \int_{-t + \lambda_{t}^{\beta}}^{t+\lambda_{t}^{\beta}} g_Z(s)ds - \int_{-t}^{-t} g_Z(s)ds = \\
	& = \int_{-t + \lambda_{t}^{\beta}}^{\lambda_{t}^{\beta}} g_Z(s)ds + \int_{\lambda_{t}^{\beta}}^{t+\lambda_{t}^{\beta}} g_Z(s)ds - 2 \int_{0}^{t} g_Z(s) ds \leq \\
	& \leq t \max_{-t+\lambda_{t}^{\beta} \leq s \leq \lambda_{t}^{\beta}} g_Z(s)  + t \max_{\lambda_{t}^{\beta} \leq s \leq t + \lambda_{t}^{\beta}} g_Z(s) - 2t \min_{0 \leq s \leq t} g_Z(s) \leq \\
	& \leq 2t \left( \max_{-t+\lambda_{t}^{\beta} \leq s \leq \lambda_{t}^{\beta}} g_Z(s) - \min_{0 \leq s \leq t} g_Z(s) \right) \leq 0.
	\end{align*}
	
	Finally, since $U(t) + V(t) = \beta$ y $U(t) \leq V(t)$ $\Rightarrow$ $U(t) \leq \beta/2 \leq V(t).$
	
	b) Consider the random variable $-Z$ which is absolutely continuous, symmetric and unimodal about $t=0.$
	Denote $F_Z$ the cumulative distribution function of $Z$ and $F_{-Z}$ the cumulative distribution function of  $-Z.$
	In addition, observe that given  $t \in \mathbb{R}$
	$$F_{-Z}(t) = P(-Z \leq t) = P(Z \geq -t) = 1 - F_Z(-t).$$
	
	\begin{align*}
	U_{-Z}(t) & = F_{-Z}(t + \lambda_{t}^{\beta}) - F_{-Z}(t) = 1 - F_{Z}(-t - \lambda_{t}^{\beta}) - \left( 1 - F_{Z}(-t) \right) = \\
	& = F_{Z}(-t) - F_{Z}(-t - \lambda_{t}^{\beta}) = V_{Z}(-t).
	\end{align*}
	Analogously,
	\begin{align*}
	V_{-Z}(t) & = F_{-Z}(t) - F_{-Z}(t- \lambda_{t}^{\beta}) = 1 - F_{Z}(-t) - \left(1 - F_{Z}(-(t-\lambda_{t}^{\beta})) \right) = \\
	& = F_{Z}(-t + \lambda_{t}^{\beta}) - F_{Z}(-t) = U_{Z}(-t).
	\end{align*}
	
	Then, if $t<0$ we have that  $-t>0$ and since part (a) of the proof holds we have that,
	$$U_{-Z}(-t) \leq \beta/2 \leq V_{-Z}(-t) \ \Rightarrow V_{Z}(t) \leq \beta/2 \leq U_{Z}(t).$$
	
\end{proof}

\begin{lemma}[\textbf{P.3.}] \label{lema2Propiedad3}
	Let $Z$ be an absolutely continuous random variable with $C^1$ cumulative distribution function, $F,$ and density $g.$
	Let $\beta \in (0,1].$ Let $t_{0} \in \mathbb{R}$ such that the density function  $g$ satisfies
	$g(t_0 - \lambda_{t_0}^{\beta}) \in supp(g)$ or $g(t_0 + \lambda_{t_0}^{\beta}) \in supp(g).$ Then, there exists an interval  $I$ centred at $t_0$ and function
	$\lambda^{\beta}:I \rightarrow \mathbb{R}_{\geq 0}$ such that $\lambda$ is $C^1$ on $I,$ $\lambda^{\beta}(t_0) = \lambda_{t_0}^{\beta}.$
	Even more, for each $s \in I,$
	$$\frac{\partial \lambda^{\beta}}{ \partial t}(s)   =
	- \frac{g(t + \lambda^{\beta}(s)) - g(t - \lambda^{\beta}(s))}{g(t + \lambda^{\beta}(s)) + g(t - \lambda^{\beta}(s))}
	$$
\end{lemma}

\begin{proof}
	The proof follows straight forward applying the implicit function theorem to the function
	$W: \mathbb{R} \times \mathbb{R}_{\geq 0} \rightarrow \mathbb{R},$

	$$W(t,\lambda) = F(t+\lambda) - F(t-\lambda) - \beta.$$

	Then we have,
	$$\frac{\partial W}{\partial t}(t,\lambda) = \frac{\partial}{\partial t} \left(F(t + \lambda) - F(t - \lambda) - \beta \right) = g(t + \lambda) - g(t-\lambda).$$
	$$\frac{\partial W}{\partial \lambda}(t,\lambda) = \frac{\partial}{\partial \lambda} \left(F(t + \lambda) - F(t - \lambda) - \beta \right) = g(t + \lambda) + g(t-\lambda).$$
	
\end{proof}

\begin{lemma}[\textbf{P.3.}] \label{lema3Propiedad3}
	Let $Z$ be an absolutely continuous, symmetric and unimodal about $t=0$ random variable, such that the cumulative distribution function $F$ is $C^1$.
	Let $\beta \in (0,1],$ and $g$ the density function such that $g(t+\lambda_{t}^{\beta})g(t-\lambda_{t}^{\beta})>0,$ which in addition satisfies that
	\begin{equation*} \label{desigualdadLema3Propiedad3}
	g(t) \geq 2 \frac{g(t+\lambda_{t}^{\beta})g(t-\lambda_{t}^{\beta})}{g(t+\lambda_{t}^{\beta})+g(t-\lambda_{t}^{\beta})} \ \forall t \in \mathbb{R}.
	\end{equation*}
	Then,
	\begin{enumerate}[a)]
		\item $LD_{S}^{\beta}(t,F)$ is non increasing if $t>0.$
		\item $LD_{S}^{\beta}(t,F)$ is non decreasing if $t<0.$
	\end{enumerate}
	
\end{lemma}

\begin{proof}
	Following Lemma \ref{lema2Propiedad3} and for the sake of simplicity denote $\lambda_{t}^{\beta} = \lambda^{\beta}(t).$
	
	It is clear  that,
	$$LD_{S}^{\beta}(t,F) = \frac{2}{\beta^2} \left[ F(t + \lambda(t)) - F(t) \right] \left[ \beta - (F(t + \lambda(t)) - F(t)) \right] =
	\frac{2}{\beta^2} U(t)(\beta - U(t)).$$
	
	The derivative of  $LD_{S}^{\beta}$ with respect to  $t$ is:
	$$\frac{\partial LD_{S}^{\beta}}{\partial t}(t,F) = \frac{2}{\beta^2} \left[ \beta \frac{\partial U}{\partial t}(t) - 2 U(t) \frac{\partial U}{\partial t}(t) \right] =
	\frac{2}{\beta^2} \frac{\partial U}{\partial t}(t) \left[ \beta - 2U(t) \right].$$
	
	By Lemma \ref{lema2Propiedad3} and considering the derivative of $U(t)$ respect to $t,$ we have that:
	\begin{align*}
	\frac{\partial U}{\partial t}(t) & = \frac{\partial}{\partial t} \left( F(t + \lambda(t)) - F(t) \right) =
	g(t + \lambda(t))\left(1 + \frac{\partial \lambda}{\partial t}(t)\right) - g(t) = \\
	& = g(t + \lambda(t)) +  g(t + \lambda(t)) \frac{\partial \lambda}{\partial t}(t) - g(t) = \\
	& = g(t + \lambda(t)) - g(t) - g(t + \lambda(t)) \frac{g(t + \lambda(t)) - g(t - \lambda(t))}{g(t + \lambda(t)) + g(t - \lambda(t))} = \\
	& = 2 \frac{g(t + \lambda(t)) g(t - \lambda(t))}{g(t + \lambda(t)) + g(t - \lambda(t))} - g(t) \leq 0.
	\end{align*}
	
	From Lemma \ref{lema1Propiedad3} we have that:
	\begin{enumerate}[a)]
		\item If $t<0,$ $\displaystyle U(t) > \frac{\beta}{2} \Rightarrow \frac{\partial LD_{S}^{\beta}}{\partial t}(t,F) = \frac{2}{\beta^2} \frac{\partial U}{\partial t}(t) \left[ \beta - 2U(t) \right] \geq 0.$
		\item if $t>0,$ $\displaystyle U(t) < \frac{\beta}{2} \Rightarrow \frac{\partial LD_{S}^{\beta}}{\partial t}(t,F) = \frac{2}{\beta^2} \frac{\partial U}{\partial t}(t) \left[ \beta - 2U(t) \right] \leq 0.$
	\end{enumerate}
	
\end{proof}

\begin{lemma}[\textbf{P.3.}] \label{lema4Propiedad3}
	Let $Z$ be a random variable and $\mu \in \mathbb{R}.$ Let $\beta \in (0,1]$ and $Y = Z - \mu.$ Denote $F_{Z}$ and $F_{Y}$ to the corresponding cumulative distribution functions, then, \\
	$LD_{S}^{\beta}(t,F_Z) = LD_{S}^{\beta}(t-\mu,F_Y).$
\end{lemma}

\begin{proof}
	Let $t \in \mathbb{R},$ we have that $F_Z(t) = F_Y(t - \mu).$
	Then,
	$$U_Z(t) = F_Z(t + \lambda_{t}^{\beta}) - F_Z(t) = F_Y(t-\mu - \lambda_{t}^{\beta}) - F_Y(t- \mu) = U_{Y}(t - \mu).$$
	entails the desired equality.
\end{proof}

Finally we prove \textbf{P.3.}

\begin{proof}[Proof: \textbf{P.3.} monotonicity relative to the deepest point]
	Let $t \in \mathbb{R}$ and $Y = Z - \theta.$
	Suppose that $t>\theta$ then $t-\theta>0.$
	On the other hand, $(1-s)\theta + st = \theta + s(t-\theta)$ and $s(t-\theta) < t-\theta.$
	Then, Lemmas \ref{lema3Propiedad3} and \ref{lema4Propiedad3} entail that,
	
	\begin{eqnarray*}
		LD_S^{\beta}\left( t,F_Z)\right)  &=& LD_S^{\beta}\left( t-\theta,F_Y\right)  \leq LD_S^{\beta} \left( s(t-\theta),F_Y \right) = \\
		&=& LD_S^{\beta} \left( s(t-\theta) + \theta,F_Z \right) = LD_S^{\beta}\left( (1-s) \theta + st,F_Z)\right) .
	\end{eqnarray*}

\end{proof}

%
%

Before proving \textbf{P.5.} the following result must be stated.
\begin{lemma} \label{contenx}
	Let $Z$ be an absolutely continuous random variable with cumulative distribution function  $F,$ then $F^{-1}(s)= \inf \{ t \in \mathbb{R} : F(t) \geq s \}.$ Let $ (z_n)_{ n \geq 1 }$ be a real sequence such that
	$z_n \xrightarrow[n \to +\infty]{} z$ and $\beta \in (0,1].$ Then, $$LD_{S}^{\beta}(z_n,F) \xrightarrow[n \to +\infty]{} LD_{S}^{\beta}(z,F).$$
\end{lemma}

\begin{proof}
	
	Since  $F$ is continuous  is enough to show that  $ \lambda_{z_n}^{\beta} \xrightarrow[n \to +\infty]{} \lambda_{z}^{\beta}.$
	
	Let $t \in \mathbb{R},$ denote $F_z(t) = \frac{1}{2}F(t) + \frac{1}{2} \left(1 - F(2z-t) \right)$ to the symmetrize version of $F$ about $z.$ Recall that,
	\begin{align*}
	\lambda_{z}^{\beta}  & :  F(z + \lambda_{z}^{\beta}) - F(z - \lambda_{z}^{\beta}) = \beta \\
	\lambda_{z_n}^{\beta}  & :  F(z_n + \lambda_{z_n}^{\beta}) - F(z_n - \lambda_{z_n}^{\beta}) = \beta \ \mbox{ for each} \ n \in \mathbb{N}
	\end{align*}
	
	Then,
	\begin{align*}
	F_{z}(z + \lambda_{z}^{\beta}) & = \frac{1}{2}F(z + \lambda_{z}^{\beta}) + \frac{1}{2} \left(1 - F(2z - (z + \lambda_{z}^{\beta})) \right) =
	\frac{1}{2}F(z + \lambda_{z}^{\beta}) + \frac{1}{2} \left(1 - F(z - \lambda_{z}^{\beta})) \right) \\
	& = \frac{1}{2} \left[ F(z + \lambda_{z}^{\beta}) - F(z - \lambda_{z}^{\beta}) \right] + \frac{1}{2} = \frac{\beta}{2} + \frac{1}{2}.
	\end{align*}
	
	Meaning that  $z + \lambda_{z}^{\beta}  = F_{z}^{-1}(\frac{\beta}{2} + \frac{1}{2}),$ analogously for
	$z_n + \lambda_{z_n}^{\beta}.$
	Given that $F_{z_n}(t) \xrightarrow[n \to +\infty]{} F_z(t)$ is clear that
	\begin{equation*}
	z_n + \lambda_{z_n}^{\beta} = F_{z_n}^{-1} \left( \frac{\beta}{2} + \frac{1}{2} \right) \xrightarrow[n \to +\infty]{} F_{z}^{-1} \left( \frac{\beta}{2} + \frac{1}{2} \right) = z + \lambda_{z}^{\beta}.
	\end{equation*}
\end{proof}

\begin{proof}[Proof: \textbf{P.5.} continuous as a function of $x$]

	It goes straight forward from Lemma \ref{contenx}  and the Dominated Convergence Theorem.
	%
\end{proof}

\begin{proof}[Proof: \textbf{P.6.} continuous as a functional of $P$]

	It goes straight forward from  Theorem 2.1, Billingsley \cite{B68} and the fact that $f:\mathbb{E} \rightarrow \mathbb{R}$ in uniformly continuous. Since  $F_{n,f}$ is the cumulative distribution function of  $f(X_n)$ which converges pointwise to  $F_f,$ then is clear that
	$$LD_{S}^{\beta}(f(x),P_{n,f}) \xrightarrow[n \to +\infty]{} LD_{S}^{\beta}(f(x),P).$$ Then it is consequence of the Dominated Convergence Theorem.

	%
	%
	%
\end{proof}

\subsection{ Uniform Strong Consistency of the IDLD.}
\setcounter{equation}{0}

In order to establish the uniform strong convergence of the one dimensional simplicial local depth. The following Lemma must be proved in advanced.

First of all, it is important to note the following facts. Assuming that the conditions stated in Remark 4 hold. For the sake of simplicity denote,  $\lambda=\lambda_{z}^{\beta},$
$p_{+} = F(z+ \lambda),$ $p_{-} = F(z- \lambda)$ and $p = F(z).$ Let $p \in (0,1),$ then,

\begin{itemize}
	
	\item[(i)] $Q_{p,n} = Z_{([np]+1)},$ $Q_{p_{+},n} = Z_{([np_{+}] + 1)}$ and $ Q_{p_{-},n} = Z_{([np_{-}] + 1)}.$
	
	\item[(ii)] $ \displaystyle F_n(Q_{p_{+},n}) - F_n(Q_{p_{-},n}) = \frac{[n p_{+}] + 1}{n} - \frac{[n p_{-}] + 1}{n} = \frac{ [n p_{+}] - [n p_{-}]}{n}.$
	
	Moreover, $ \displaystyle \frac{ [n p_{+} - n p_{-}]}{n} \leq \frac{ [n p_{+}] - [n p_{-}]}{n} \leq \frac{ [n p_{+}] - [n p_{-}] + 1}{n}.$
	
	$[n p_{+} - n p_{-}] = [n (p_{+} - p_{-})] = [n \beta].$ Then,
	
	$ \displaystyle \frac{[n \beta]}{n} \leq F_n(Q_{p_{+},n}) - F_n(Q_{p_{-},n}) \leq  \frac{[n \beta] + 1}{n}.$
	
	\item[(iii)]  $Z_{( [np_{-} ] + 1)} \leq z \leq Z_{([np_{+}]+1)}.$
	
	\item[(iv)]  $ [z - d^{(k)}(z),z + d^{(k)}(z)] \subset  [ Z_{ ([np_{-}] + 1)} , Z_{ ([np_{+}]+1)} ].$
	
	\item[(v)] $d_{(k)}(z) = \min \{z - Q_{p_{-},n}, z + Q_{p_{+},n} \}.$
	
	\item[(vi)] $ \displaystyle F_n(Q_{p_{+},n}) - F_n(z + d_{(k)}(z)) \leq \frac{1}{n}$ and $ \displaystyle F_n(Q_{p_{-},n}) - F_n(z - d_{(k)}(z)) \leq \frac{1}{n}.$
	
	\item[(vii)] $ \beta(k) \leq \beta \leq \beta(k) + 1.$
	
\end{itemize}

\begin{lemma} Let $Z$ be an absolutely continuous random variable with distribution $P^1.$ Suppose given  $Z_1, \dots, Z_n$  iid random variables, also with distribution $P^1$.
	Let $z_p = F^{-1}(p)$ be the quantile $p \in (0,1)$ from $F$ and $Q_{p,n}$ the quantile $p$ from $F_n, $ which is the empirical cumulative distribution function of  $Z_1, \dots, Z_n.$  Then,
	
	\begin{itemize}
		\item[(i)] $Q_{p,n} = Z_{ \left([np] +1 \right) }.$
		\item[(ii)] $| F_n(Q_{p,n}) - F(z_p) | \leq \frac{1}{n} \ \forall \ p \in (0,1). $
		\item[(iii)] $ | F(Q_{p,n}) - F(z_p) | \leq ||F_n - F ||_{\infty} + \frac{1}{n}.$
	\end{itemize}
	
\end{lemma}

\begin{proof}[Proof: Lemma 1]
	
	\begin{enumerate}[(i)]
		\item It follows straight forward by definition.
		
		\item Let $p \in (0,1),$
		\begin{equation*}
		| F_n(Q_{p,n}) - F(z_p) | = | F_n(Q_{p,n}) - p | = \frac{[np] + 1}{n} - p = \frac{[np] - np + 1}{n} \leq \frac{1}{n}.
		\end{equation*}
		
		\item Let $p \in (0,1),$
		\begin{align}
		| F(Q_{p,n}) - F(z_p) | & \leq | F(Q_{p,n}) - F_n(Q_{p,n}) | + |F_n(Q_{p,n}) - F(z_p) | \nonumber \\
		& \leq \sup_{t \in \mathbb{R}} |F_n(t) - F(t) | + \frac{1}{n} = ||F_n - F ||_{\infty} + \frac{1}{n}. \nonumber
		\end{align}
		
	\end{enumerate}
\end{proof}

\begin{lemma} \label{desigualdadLDS}
	Let $Z_1, \dots, Z_n$ be a real random sample with cumulative distribution function $F.$ Let $\beta \in (0,1]$ and $z \in \mathbb{R}.$ Then,
	\begin{equation}
	\left| ELD_{S}^{\beta(k)}(z,F_n) - LD_{S}^{\beta}(z,F) \right| \leq \frac{1}{2} \left( 1 - \left( \frac{\beta(k)}{\beta} \right)^2 \right) + \frac{2}{\beta^2} \left(\frac{8}{n} + 4 ||F_n - F||_{\infty} \right)
	\label{desguniv}
	\end{equation}
\end{lemma}

\begin{proof}[Proof: Lemma 2]
	
	For the sake of simplicity denote $\lambda=\lambda_{z}^{\beta}$ and
	$d^k=d^{(k)}(z).$
	\par
	\begin{align}
	&  \Big| \left( F(z + \lambda) - F(z) \right) \left( F(z) - F(z - \lambda) \right) - \left( F_n(z + d^k) - F_n(z) \right) \left( F_n(z) - F_n(z - d^k) \right) \Big| = \nonumber \\
	&  =  \Big|  \left[ F(z + \lambda)F(z) - F(z + \lambda)F(z - \lambda) - F(z)^2 + F(z)F(z - \lambda) \right] - \nonumber \\
	&   -  \left[ F_n(z + d^k)F_n(z) - F_n(z + d^k)F_n(z - d^k) - F_n(z)^2 + F_n(z)F_n(z - d^k) \right] \Big|  = \nonumber \\
	&   = \Big| F(z + \lambda)F(z) - F(z + \lambda)F(z - \lambda) - F(z)^2 + F(z)F(z - \lambda) - \nonumber \\
	&   -  F_n(z + d^k)F_n(z) + F_n(z + d^k)F_n(z - d^k) + F_n(z)^2 - F_n(z)F_n(z - d^k) \Big| = \nonumber \\
	&  \leq  \Big| F(z+\lambda)F(z) - F_n(z+d^k)F_n(z) \Big| + \Big| F_n(z+d^k)F_n(z-d^k) - F(z+\lambda)F(z-\lambda)  \Big| + \label{4term} \\
	&  +  \Big| F_n(z)^2 - F(z)^2  \Big| + \Big| F(z-\lambda)F(z) - F_n(z-d^k)F_n(z)  \Big| \nonumber
	\end{align}
	
	We analyze each term of Equation (\ref{4term}),

	\begin{enumerate}[(a)]
		
		\item
		\begin{align*}
		\displaystyle
		& \Big| F(z+\lambda)F(z) - F_n(z+d^k)F_n(z) \Big| = \\
		& = \Big| F(z+\lambda)F(z) - F(z)F_n(Q_{p_{+},n}) + F(z)F_n(Q_{p_{+},n}) - F_n(z+d^k)F_n(z) \Big| \leq \\
		& \leq F(z) \Big| F(z+\lambda) - F_n(Q_{p_{+},n}) \Big| + \\
		& + \Big| F(z)F_n(Q_{p_{+},n}) - F(z)F_n(z+d^k) + F(z)F_n(z+d^k)-F_n(z+d^k)F_n(z) \Big| \leq \\
		& \leq \frac{1}{n} + F(z) \Big| F_n(Q_{p_{+},n}) - F_n(z+d^k)\Big| + \Big| F(z) - F_n(z) \Big| F_n(z+d^k) \leq \\
		& \leq \frac{1}{n} + \frac{1}{n} + \|F - F_n\|_{\infty} = \frac{2}{n} + \|F - F_n\|_{\infty}.
		\end{align*}

		\item
		\begin{align*}
		\displaystyle
		& \Big| F_n(z+d^k)F_n(z-d^k) - F(z+\lambda)F(z-\lambda)\Big| = \\
		& = \Big| F_n(z+d^k)F_n(z-d^k) - F(z+\lambda)F_n(z-d^k) + \\
		& + F(z+\lambda)F_n(z-d^k) - F(z+\lambda)F(z-\lambda)\Big| \leq \\
		& \leq F_n(z-d^k) \Big| F_n(z+d^k) - F(z+\lambda) \Big| + F(z+\lambda) \Big| F_n(z-d^k) - F(z-\lambda) \Big| \leq \\
		& \leq \Big| F_n(z+d^k) - F(z+\lambda) \Big| + \Big| F_n(z-d^k) - F(z-\lambda) \Big| \leq \\
		& \leq \Big| F_n(z+d^k) - F_n(Q_{p_{+},n}) \Big| + \Big| F_n(Q_{p_{+},n}) - F(z+\lambda) \Big| + \\
		& + \Big| F_n(Q_{p_{-},n}) - F_n(z-d^k) \Big| + \Big| F_n(Q_{p_{-},n}) - F(z-\lambda)\Big| \leq
		\frac{1}{n} + \frac{1}{n} + \frac{1}{n} + \frac{1}{n} = \frac{4}{n}.
		\end{align*}
		
		\item
		\begin{align*}
		& \Big| F_n(z)^2 - F(z)^2 \Big| = \Big| F_n(z) - F(z)\Big| \mbox{   } \Big| F_n(z) + F(z)\Big| \leq \\
		& \leq 2 \Big|F_n(z) - F(z) \Big| \leq 2 \|F_n - F\|_{\infty}.
		\end{align*}
		
		\item Analogue to item (a).
	\end{enumerate}
	Finally, denote
	$$ H = \left( F(z + \lambda)F(z) \right) \left( F(z)F(z-\lambda) \right)$$
	
	and
	
	$$G = \left( F_n(z+d^k) - F_n(z) \right) \left( F_n(z) - F_n(z-\lambda) \right).$$
	
	Then,
	\begin{align*}
	& \Big| ELD_{S}^{\beta}(z,F_n) - LD_{S}^{\beta}(z,F) \Big| = \left| \frac{2}{\beta(k)^2}G - \frac{2}{\beta^2}H \right| \leq
	\left| \frac{2}{\beta(k)^2}G - \frac{2}{\beta^2}G \right| + \left| \frac{2}{\beta^2}G - \frac{2}{\beta^2}H \right| \leq \\
	& \leq \left( \frac{2}{\beta(k)^2} - \frac{2}{\beta^2} \right) | G |  + \frac{2}{\beta^2} \Big| G-H \Big|.
	\end{align*}
	
	On one hand, since Proposition 2 holds its clear that each term of  $G$ is smaller than or equal to  $\frac{\beta(k)^2}{2}$. Hence,
	
	\begin{equation}
	\left( \frac{2}{\beta(k)^2} - \frac{2}{\beta^2} \right) |G| \leq \left( \frac{2}{\beta(k)^2} - \frac{2}{\beta^2} \right) \frac{\beta(k)^2}{4} =
	\frac{1}{2} \left( 1 - \left( \frac{\beta(k)}{\beta} \right)^2 \right).
	\label{desg1}
	\end{equation}
	
	On the other hand, we already know that,
	\begin{equation}
	\Big|G-H\Big| \leq \frac{8}{n} + 4\|F - F_n\|_{\infty}.
	\label{desg2}
	\end{equation}
	
	From Inequalities (\ref{desg1}) and (\ref{desg2}) we prove the inequality stated in the statement.
\end{proof}

\begin{proof}[Proof:Theorem 1]
	
	\textit{(a)} Let $f \in \mathbb{E}'$ and $x \in \mathbb{E}.$ Denote $P_{f}$ to the probability measure associated to  $f(X)$ where $X$ is a random element on  $\mathbb{E}$ with probability measure $P.$ Analogously, denote $P_{n,f}$ to the empirical probability measure of  $P_f$ based on  $f(X_1), \dots, f(X_n) $ and $F_{f,n}$ to the empirical cumulative distribution function.
	
	By Proposition 2 we have,
	\begin{equation*}
	\Big| ELD_{S}^{\beta(k)}(f(x),P_{f,n}) - LD_{S}^{\beta}(f(x),P_f)\Big| \leq \frac{1}{2} \left( 1 - \left( \frac{\beta(k)}{\beta} \right)^2 \right) + \frac{2}{\beta^2} \left(\frac{8}{n} + 4 ||F_{n,f} - F_{f}||_{\infty} \right).
	\end{equation*}
	
	Observe that,
	\begin{align}
	\frac{1}{2} \left( 1 - \left( \frac{\beta(k)}{\beta} \right)^2 \right) = \frac{1}{2} \frac{ \beta^2 - \beta(k)^2 }{\beta^2} =
	\frac{1}{2} (\beta - \beta(k) ) \frac{  (\beta + \beta(k) ) }{\beta^2} \leq  \frac{2}{2n\beta^2} = \frac{1}{n \beta^2}.
	\nonumber
	\end{align}
	
	Thus,
	\begin{equation}
	\Big|ELD_{S}^{\beta(k)}(f(x),P_{f,n}) - LD_{S}^{\beta}(f(x),P_f)\Big| \leq \frac{1}{\beta^2} \left(\frac{17}{n} + 8 ||F_{n,f} - F_{f}||_{\infty} \right).
	\label{xfijo}
	\end{equation}
	Since it does not depend on $x$ the inequality hold for the supreme of the left hand side of Inequality (\ref{xfijo}).
	
	\begin{equation} \label{supremonodependedex}
	\sup_{x \in \mathbb{E}} \Big| ELD_{S}^{\beta(k)}(f(x),P_{f,n}) - LD_{S}^{\beta}(f(x),P_f) \Big| \leq \frac{1}{\beta^2} \left(\frac{17}{n} + 8 ||F_{n,f} - F_{f}||_{\infty} \right).
	\end{equation}
	
	\bigskip
	
	Thus,
	
	\begin{align} \label{supremonodependedex}
	& E \left[ \sup_{x \in \mathbb{E}} \Big| ELD_{S}^{\beta(k)}\left(f(x),P_{f,n}\right) - LD_{S}^{\beta}\left(f(x),P_f\right) \Big| \right] \leq
	E \left[  \frac{1}{\beta^2} \left(\frac{17}{n} + 8 ||F_{f,n} - F_{f}||_{\infty} \right) \right] =  \nonumber \\
	& = \frac{1}{\beta^2} \left( \frac{17}{n} + 8 E \left[ ||F_{f,n} - F_{f}||_{\infty} \right] \right).
	\end{align}

	From  (\ref{supremonodependedex}) it is enough to show that
	$ E \left[ ||F_{f,n} - F_{f}||_{\infty} \right]  \xrightarrow[n \to +\infty]{}0.$
	
	Let $\epsilon > 0,$ there exists  $n_0 \in \mathbb{N}$ such that $ \displaystyle \forall \ n \geq n_0,$ $4 e^{-n \frac{ \epsilon^2}{2} } < \frac{\epsilon}{2};$
	
	\begin{align*}
	& E \left[ ||F_{f,n} - F_{f}||_{\infty} \right] = \\
	& = E \left[ ||F_{f,n} - F_{f}||_{\infty} \mathcal{I} \Big\{ ||F_{f,n} - F_{f}||_{\infty} \leq \frac{\epsilon}{2} \Big\} \right] +
	E \left[ ||F_{f,n} - F_{f}||_{\infty} \mathcal{I} \Big\{ ||F_{f,n} - F_{f}||_{\infty} > \frac{\epsilon}{2} \Big\} \right] \leq \\
	& \leq \epsilon \mathbb{P} \left( ||F_{f,n} - F_{f}||_{\infty} \leq \frac{\epsilon}{2} \right) +
	2 \mathbb{P} \left( ||F_{f,n} - F_{f}||_{\infty} > \frac{\epsilon}{2} \right) \leq \\
	& \leq \frac{\epsilon}{2} + 2 \mathbb{P} \left( ||F_{f,n} - F_{f}||_{\infty} > \frac{\epsilon}{2} \right) .
	\end{align*}

	By Dvoretzky-Kiefer-Wolfowitz inequality Massart \cite{M90} we have that
	
	\begin{equation} \label{DKV}
	E \left[ ||F_{f,n} - F_{f}||_{\infty} \right] \leq \frac{\epsilon}{2} + 4 e^{-n \frac{\epsilon^2}{2}} < \epsilon
	\end{equation}
	
	It is important to note that the convergence does not depend on the functional $f;$ which will be useful to prove part (b) of the result.

	\bigskip
	
	\textit{(b)} It follows straight forward from  part (a) of the theorem and the fact that it is the integral of a mensurable, positive and bounded  function. Hence, given $\epsilon >0$ there exists $n_1 \in \mathbb{N}$ such that for every $n>n_1,$
	
	$$
	E \left[ \sup_{x \in \mathbb{E}} \left| ELD_S^{\beta(k)}(f(x),P_{f,n})-LD_S^{\beta}(f(x),P_{f}) \right| \right] < \epsilon.
	$$
	
	Then,
	
	\begin{align}
	& E \left[ \sup_{x \in \mathbb{E}} \Big| EIDLD^{\beta(k)}(x,P_n) - IDLD^{\beta}(x,P) \Big| \right] \leq  \nonumber  \\
	&  \leq E \left[ \sup_{x \in \mathbb{E}} \int  \Big|ELD_{S}^{\beta(k)}(f(x),P_{f,n}) - LD_{S}^{\beta}(f(x),P_f)\Big| dQ(f) \right] = \nonumber  \\
	&  = E \left[ \int \sup_{x \in \mathbb{E}} \Big|ELD_{S}^{\beta(k)}(f(x),P_{f,n}) - LD_{S}^{\beta}(f(x),P_f)\Big| dQ(f) \right] = \nonumber  \\
	&  = \int E \left[ \sup_{x \in \mathbb{E}} \Big|ELD_{S}^{\beta(k)}(f(x),P_{f,n}) - LD_{S}^{\beta}(f(x),P_f)\Big| \right] dQ(f) \leq \nonumber  \\
	&  = \int  \epsilon \ dQ(f) =  \epsilon \ \mbox{ if } \ n > n_1. \nonumber
	\end{align}
	
\end{proof}

\begin{proof}[Proof:Theorem 2]

	Note that,
	\begin{align}
	& \mathbb{P} \left( \sup_{x \in \mathbb{E}} \Big| EIDLD^{\beta}(x,P_n) - IDLD^{\beta}(x,P) \Big|  \xrightarrow[n \to +\infty]{} 0 \right) =   \nonumber \\
	& = \mathbb{P} \left( \bigcap_{ \epsilon > 0} \bigcup_{n \in \mathbb{N} } \bigcap_{l \geq n} \left\{ \sup_{x \in \mathbb{E}} \Big| EIDLD^{\beta(k)}(x,P_l) - IDLD^{\beta}(x,P) \Big| < \epsilon \right\} \right) =  \nonumber  \\
	& = 1 - \mathbb{P} \left( \bigcup_{ \epsilon > 0} \bigcap_{n \in \mathbb{N} } \bigcup_{l \geq n} \left\{ \sup_{x \in \mathbb{E}} \Big| EIDLD^{\beta(k)}(x,P_l) - IDLD^{\beta}(x,P) \Big| > \epsilon \right\} \right). \nonumber
	\end{align}
	
	It is enough to show that
	$$ \mathbb{P} \left( \bigcup_{ \epsilon > 0} \bigcap_{n \in \mathbb{N} } \bigcup_{l \geq n} \left\{ \sup_{x \in \mathbb{E}} | EIDLD^{\beta(k)}(x,P_l) - IDLD^{\beta}(x,P) | > \epsilon \right\} \right)=0.$$
	
	By Borell-Cantelli lemma it is enough to prove that if the probability of the sets
	$$ A_n = \left\{ \sup_{x \in \mathbb{E}} |EIDLD^{\beta(k)}(x,P_n) - IDLD^{\beta}(x,P) | > \epsilon \right\},$$
	are summable, then, for all $\epsilon >0, \\
	\mathbb{P} \left(  \bigcap_{n \in \mathbb{N} } \bigcup_{l \geq n} \left\{ \sup_{x \in \mathbb{E}} \Big| EIDLD^{\beta(k)}(x,P_l) - IDLD^{\beta}(x,P) \Big| > \epsilon \right\} \right)=0 $ and the prove would be done.
	
	\bigskip
	
	Let $\epsilon > 0,$
	\begin{align*}
	& \sup_{x \in \mathbb{E}} \Big| EIDLD^{\beta(k)}(x,P_n) - IDLD^{\beta}(x,P) \Big| = \sup_{x \in \mathbb{E}} \Big| \int ELD^{\beta(k)}_S (f(x),P_{n,f}) - LD^{\beta}_S (f(x),P_f) dQ \Big| \leq \\
	& \leq \sup_{x \in \mathbb{E}} \int \Big| ELD^{\beta(k)}_S (f(x),P_{n,f}) - LD^{\beta}_S (f(x),P_f) \Big| dQ = \\
	& = \int \sup_{x \in \mathbb{E}} \Big| ELD^{\beta(k)}_S (f(x),P_{n,f}) - LD^{\beta}_S (f(x),P_f) \Big| dQ \leq \\
	& \leq \int \frac{1}{\beta^2} \left(\frac{8}{n} + 4 ||F_{n,f} - F_{f}||_{\infty} \right) dQ \
	= \frac{1}{\beta^2}\frac{8}{n} +  \frac{1}{\beta^2} 4 \int ||F_{n,f} - F_{f}||_{\infty} dQ \leq \\
	& \leq \frac{8}{n \beta^2} + \frac{1}{2 \beta^2} \sup_{f \in \mathbb{E}^{'}} \|F_{n,f} - F_{f}\|_{\infty}.
	\end{align*}
	
	Given that $ \displaystyle \sup_{f \in \mathbb{E}'} ||F_{n,f} - F_{f}||_{\infty} < + \infty $ there exists  $f_0 \in \mathbb{E}'$ such that
	$$\sup_{f \in \mathbb{E}'} ||F_{n,f} - F_{f}||_{\infty} \leq ||F_{n,f_0} - F_{f_0}||_{\infty} + \beta^{2} \epsilon, $$
	then
	$$ \frac{8}{n \beta^2} + \frac{1}{2 \beta^2} \sup_{f \in \mathbb{E}^{'}} \|F_{n,f} - F_{f}\|_{\infty} \leq
	\frac{8}{n \beta^2} + \frac{1}{2 \beta^2} ||F_{n,f_0} - F_{f_0}||_{\infty} + \frac{\beta^2 \epsilon}{2 \beta^2}.$$

	By  Dvoretzky-Kiefer-Wolfowitz inequality,
	
	\begin{align*}
	& \mathbb{P}(A_n) \leq \mathbb{P} \left( \frac{8}{n \beta^2} + \frac{1}{2 \beta^2} ||F_{n,f_0} - F_{f_0}||_{\infty} + \frac{\epsilon}{2} > \epsilon \right) =
	\mathbb{P} \left( ||F_{n,f_0} - F_{f_0}||_{\infty} > \epsilon \beta^2 - \frac{16}{n} \right) \leq \\
	& \leq 2 \exp \left \{ -2n \left( \epsilon \beta^2 - \frac{16}{n} \right)^2 \right \}.
	\end{align*}
	
	Which is bounded by Borell-Cantelli's lemma,
	\begin{align}
	\sum_{n \in \mathbb{N}} \mathbb{P}(A_n) \leq 2 \sum_{n \in \mathbb{N}} \exp \left \{ -2n \left( \epsilon \beta^2 - \frac{16}{n} \right)^2 \right \} < + \infty.
	\end{align}
	
\end{proof}

\subsection{Computational Cost}
It is well known that one of the main drawback of depth and local depth measures is that they are highly demandant computationally. Therefore, we analyze this problem from this perspective, focusing in the multivariate case where these problems usually arise.
In what follows we compare the computational times for the three local depths measures, LDS (Agostinelli and Romanazzi, 2011), LDPV (Paindavaine and Van Bever, 2013) and IDLD, our proposal.
We generate data that has a three group structure each is generated according to a  multivariate $p$-dimensional normal distribution, $N(\mu_i, \Sigma),$  for $i=1,2,3,$ with centers $(-3,-3,0,\dots,0),(0,0,0,\dots,0),(3,3,0,\dots,0), $ and the covariance matrix is the identity matrix, the first two variables are informative while the remaining, $p-2,$ ($p=5,35,65$) have normal independent noise variables centered at the origin with unit standard deviation.
Also we considered different sample sizes, $n=300, 2100, 3900$ and $5700.$ For ILDL $50,$ random directions were generated.
Since the computational time increases exponentially as the dimension increases and the benchmark procedures are expesive computationally, we only performed  $M=50$ replicates under each scenario.

\begin{table}[!t]
\centering
	\label{compmultiv}\par
	
		\begin{tabular}{r|crrrr}
			&       & $n=300$  & $n=2100$ & $n=3900$ & $n=5700$   \\[3pt] \hline
			$p=5$ &  LDS  & $0.785$ & $38.27$ & $131.65$ & $280.66$ \\
			&  LDPV & $4.236$ & $100.08$ & $292.67$ & $624.91$ \\
			&  IDLD & $0.397$ & $20.74$ & $73.74$ & $160.43 $  \\ \hline
			$p=35$&  LDS  & $1.770$ & $86.88$ & $299.03$ & $638.38$ \\
			&  LDPV & $7.840$ & $200.94$ & $629.07$ & $1363.97$ \\
			&  IDLD & $0.402$ & $20.68$ & $74.41$ & $160.29 $ \\ \hline
			$p=65$&  LDS  & $3.788$ & $184.92$ & $641.01$ &  $1368.31$  \\
			&  LDPV & $10.934$ & $288.79$ & $982.79$ & $2094.89$ \\
			&  IDLD & $0.406$ & $20.66$ & $75.07$ & $164.40$ \\ \hline
	\end{tabular}
	\caption{Mean computer time for LDS, LDPV and IDLD.}
\end{table}

From Table \ref{compmultiv} we can see that in every case IDLD is the fastest procedure, moreover it is not affected by the dimension of the dataset, while the computational efforts required by LDS and LDPV grow dramatically as  $p$ increases. LDPV is overall the slowest procedure. Even though all the procedures demand more time as the sample size grows, IDLD is the one with the least pronounced growth rate. In every case the mean square error between the estimates is smaller than $0.09,$ being the three proposals able to detect the multimodal structure.


\section*{Acknowledgements}
This work was  partially  supported by the Spanish Agencia Estatal de Investigaci\'on (AEI) and Fondo Europeo de
Desarrollo Regional (FEDER), Grant CTM2016-79741-R for MICROAIPOLAR project and by Grant \textsc{pict} 2018-00740 from \textsc{anpcyt} (M. Svarc). We appreciate the valuable comments done by Mat\'{i}as Fernandez Piana, in the GDP data set analysis.


\end{document}